\newif\iflong
\newif\ifshort

\iflong 
\else
\shorttrue
\fi

\documentclass{article}
\usepackage{geometry}
\usepackage{amsmath,amssymb,amsfonts}
\usepackage{balance} 

\newcommand{\mytitle}{Optimal Capacity Modification for Many-To-One Matching Problems}
\newcommand{\appendixtitle}{Supplementary Material for the Paper ``\mytitle''}

\title{\mytitle}
\author{Jiehua Chen\textsuperscript{\rm 1} \and Gergely Cs{\'a}ji\textsuperscript{\rm 2}\textsuperscript{\rm 3}}

\date{
\textsuperscript{\rm 1}TU Wien, Austria\\
\textsuperscript{\rm 2}HUN-REN Centre for Economic and Regional Studies, Hungary\\
\textsuperscript{\rm 3}ELTE Eötvös Loránd University, Hungary
}
 
\usepackage[utf8]{inputenc}
\usepackage{amsmath,amsthm}
\usepackage{graphicx,dsfont,mathtools}

\usepackage{xspace}

\usepackage{hyperref}
\hypersetup{%
 backref=true, 
 pagebackref=true, 
 hypertexnames=true,
 colorlinks=true,citecolor=green!35!black,linkcolor=red!60!black%
} 

\usepackage{xifthen}
\usepackage{paralist}

\usepackage{tikz}
\usetikzlibrary{decorations,arrows,arrows.meta,petri,topaths,backgrounds,shapes,positioning,fit,calc,decorations.pathreplacing,patterns,intersections}

\pgfdeclarelayer{bg}    
\pgfdeclarelayer{fg}    
\pgfsetlayers{bg,main,fg}  
\usepackage{algorithm,algorithmic}

\usepackage[textsize=tiny,textwidth=1.5cm,linecolor=green!70!black, backgroundcolor=green!10, bordercolor=black]{todonotes}

\usepackage[capitalize]{cleveref}

\newtheorem{lemma}{Lemma}

\newtheorem{example}{Example}
\newtheorem{remark}{Remark}
\newtheorem{theorem}{Theorem}
\newtheorem{claim}{Claim}[theorem]

\newtheorem{corollary}{Corollary}
 \theoremstyle{definition}

\newtheorem{observation}{Observation}

\crefname{table}{Table}{Tables}
\crefname{figure}{Figure}{Figures}
\crefname{theorem}{Theorem}{Theorems}
\crefname{corollary}{Corollary}{Corollaries}
\crefname{observation}{Observation}{Observations}
\crefname{lemma}{Lemma}{Lemmas}
\crefname{example}{Example}{Examples}
\crefname{reduction}{Reduction}{Reductions}
\crefname{construction}{Construction}{Constructions}
\crefname{subsection}{Subsection}{Subsections}
\crefname{section}{Section}{Sections}
\crefname{claim}{Claim}{Claims}
\crefname{clm}{Claim}{Claims}
\crefname{algorithm}{Algorithm}{Algorithm}
\crefname{definition}{Definition}{Definitions}

\newcommand{\decprob}[3]{
   \begin{center}%
     \begin{itemize}[d]
       \item[\textsc{#1}]
       \item[\textbf{Input:}]  #2\\[0.2ex]
       \item[\textbf{Question:}]  #3
     \end{itemize}
  \end{center}
}

\newcommand{\optprob}[3]{
   \begin{center}%
    \begin{minipage}{0.96\linewidth}%
      \textsc{#1}\\[0.2ex]
      \textbf{Input:} #2\\[0.2ex]
      \textbf{Task:} #3
    \end{minipage}%
  \end{center}
}

\newcommand{\rrsat}{\textsc{(2,2)\nobreakdash-$3$SAT}\xspace}

\newcommand{\MCC}{\textsc{Mul\-ti-Col\-ored Clique}}
\newcommand{\MCCs}{\textsc{MCC}}

\newcommand{\scov}{\textsc{Set Cover}\xspace}

\newcommand{\VC}{\textsc{Vertex Cover}\xspace}

\newcommand{\manyone}{\textsc{Many-To-One Matching}\xspace}
\newcommand{\manyones}{\textsc{MM}\xspace}
\newcommand{\sumcol}[1]{{\color{red!50!black}#1}}
\newcommand{\maxcol}[1]{{\color{blue!50!black}#1}}
\newcommand{\minsummatching}[1][$\Pi$]{\textsc{Min{\sumcol{Sum}} Cap #1}\xspace}
\newcommand{\minmaxmatching}[1][$\Pi$]{\textsc{Min{\maxcol{Max}} Cap #1}\xspace}
\newcommand{\sumsta}{\textsc{{\minsum}SP}}

\newcommand{\sumstud}{\textsc{{\minsum}SSP}}
\newcommand{\maxstud}{\textsc{{\minmax}SSP}}
\newcommand{\sumse}{\textsc{{\minsum}SE}}

\newcommand{\maxsta}{\textsc{{\minmax}SP}}
\newcommand{\maxse}{\textsc{{\minmax}SE}}

\newcommand{\minsum}{\textsc{Min\-\sumcol{Sum}}}
\newcommand{\minmax}{\textsc{Min\-\maxcol{Max}}}
\newcommand{\avgrank}{\textsc{AvgRank}}
\newcommand{\maxsm}{\textsc{CarSM}}
\newcommand{\lone}[1][$\pr$]{\ensuremath{|#1|_{\sumcol{1}}}}
\newcommand{\lmax}[1][$\pr$]{\ensuremath{|#1|_{\maxcol{\infty}}}}

\newcommand{\Unas}{\ensuremath{U_{\mathsf{un}}}}
\newcommand{\As}{\ensuremath{U_{\mathsf{as}}}}

\newcommand{\studlen}{\ensuremath{\Delta_{u}}}
\newcommand{\scholen}{\ensuremath{\Delta_{\!sc}}}

\newcommand{\enn}{\ensuremath{\hat{n}}}
\newcommand{\emm}{\ensuremath{\hat{m}}}

\newcommand{\m}{\hat{m}}
\newcommand{\en}{\hat{n}}
\newcommand{\vect}[1]{\ensuremath{\boldsymbol{#1}}}

\newcommand{\px}{\ensuremath{\vect{x}}}
\newcommand{\py}{\ensuremath{\vect{y}}}

\newcommand{\pr}{\ensuremath{\vect{r}}}
\newcommand{\pv}{\ensuremath{\vect{v}}}
\newcommand{\pq}{\ensuremath{\vect{q}}}

\newcommand{\quota}{\ensuremath{\boldsymbol{q}}}
\newcommand{\acset}{\ensuremath{\mathsf{A}}}
\newcommand{\instance}{\ensuremath{(U, W, (\succ_x)_{x\in U \cup W}, \quota)}}
\newcommand{\profile}{\ensuremath{U, W, (\succ_x)_{x\in U \cup W}}}
\newcommand{\sumcap}{\ensuremath{k^{+}}}
\newcommand{\maxcap}{\ensuremath{k^{\max}}}
\newcommand{\Vecs}{\ensuremath{\mathcal{V}}}
\newcommand{\JE}{\ensuremath{\mathsf{F}}}
\newcommand{\je}{\ensuremath{\mathsf{n_{\mathsf{F}}}}}
\newcommand{\stuopt}{\ensuremath{\hat{\mu}}}

\newcommand{\seqq}[1]{\ensuremath{\langle #1 \rangle}}

\newcommand{\myemph}[1]{{\color{purple!30!black}\emph{#1}}}

\definecolor{darkgreen}{rgb}{0.01,0.6,0.1}
\definecolor{darkblue}{rgb}{0,0,0.4}
\definecolor{winered}{rgb}{0.6,0.1,0.1}
\definecolor{doncolor}{RGB}{78,154,0}
\definecolor{falsecolor}{RGB}{0,55,255}
\definecolor{truecolor}{RGB}{164,0,0}
\definecolor{lightblue}{rgb}{0.527,0.805,0.977}

\newcommand{\blu}{\textcolor{black}}

\tikzset{ttrue/.style={color=truecolor!50!white}}
\tikzset{ffalse/.style={color=falsecolor!50!white}}
\tikzset{dont/.style={color=doncolor!50!white}}
\tikzset{trueline/.style =   {line width= 3pt, ttrue}}
\tikzset{falseline/.style =   {line width= 3pt, ffalse}}
\tikzset{dontline/.style =   {line width= 3pt, dont}}

\begin{document}
\maketitle
\begin{abstract}
  We consider many-to-one matching problems, where one side consists of students and the other side of schools with capacity constraints. 
  We study how to optimally increase the capacities of the schools so as to obtain
  a stable matching that is also perfect (i.e., every student is matched) or
Pareto-efficient or student-popular (i.e., there is no matching, where more students improve than get worse).
  We consider two common optimality criteria, one aiming to minimize the sum of capacity increases of all schools (abbrv.\ as \minsum) and the other aiming to minimize the maximum capacity increase of any school (abbrv.\ as \minmax).
  We obtain a complete picture in terms of computational complexity:
 In the case of stable and perfect matchings and stable and student-popular matchings using the \minmax\ criteria
  the problem is polynomial-time solvable,
  but all other remaining problems are NP-hard. 
  We further investigate the parameterized complexity and approximability and find that
  achieving stable and efficient matchings via minimal capacity increases is much harder than
  achieving stable and perfect matchings. \blu{ Finally we consider the variants of these problems, where decreases in the capacities are also allowed and obtain similar results.}
\end{abstract}




\pagestyle{plain}


\section{Introduction}\label{sec:intro}

Many-to-one matching with two-sided preferences has various real-world applications 
such as school choice, i.e., placement of students to schools~\cite{SoenmezAbdul2003}, college and university admission ~\cite{GaleShapley1962,Romero1998}, hospital/residents programs~\cite{Manlove2016}, staff recruitment~\cite{RP99,BaBa2004studentadmission}, or fair allocation in healthcare ~\cite{PSY2020healthcare}.
In general, we are given two disjoint sets of agents, $U$ and $W$, 
such that each agent has a preference list over some members of the other set, and each agent in $W$ has a capacity constraint (aka.\ a quota) which specifies how many agents from $U$ can be matched to it.
The goal is to find a \emph{good} matching (or assignment) between $U$ and $W$ without violating the capacity constraints.
For school choice and university admission, for example, the agents in $U$ would be students or high-school graduates, while the agents in $W$ would be schools and universities, respectively. 
To unify the terminology, we consider the school choice application, and call the agents in $U$ the students and the agents in $W$ the schools.

As to what defines a good matching, the answer varies from application to application.
The arguably most prominent and well-known concept is that of \myemph{stable} matching~\cite{GaleShapley1962,GusfieldIrving1989,Manlove2013}, which ensures that no two agents form a blocking pair, i.e., they do not prefer to be matched with each other over their assigned partners.
Stability is a key desideratum and has been a standard criterion for many matching applications.
On the other hand, the simplest concept is to ensure that every student is matched, and we call such matching a \myemph{perfect} matching.
Note that having a perfect matching is particularly important in school choice or university admission since every student should at least be admitted to some school/university. 
A \myemph{Pareto-efficient} (abbrv.\ \myemph{efficient}) matching ensures that no other matching can make one student better off without making another student worse off\footnote{Not to be confused with Pareto-optimality which requires that no matching exists that can make an agent (either student or school) better off without making another agent worse off.}~\cite{Ergin2002,SoenmezAbdul2003}.
Efficiency is very desirable for the students since it saves them from trying to find a better solution. A somewhat stronger restriction than Pareto-efficiency we can impose is \myemph{student-popularity}. We say that a matching $\mu$ is \myemph{student-popular}, if there is no other matching $\mu '$, such that strictly more students prefer their allocation in $\mu '$ than in $\mu$.

Stability and efficiency, even though equally desirable, are not compatible with each other (i.e., they may not be satisfiable simultaneously).
Neither is stability compatible with student-popularity or perfectness. In the worst case, a stable matching may have just half the size of the largest matching.
But what if we can modify the capacities of the schools?
\noindent To illustrate, consider the following example with five students $u_1, \ldots$, $u_5$ and three schools~$w_1,w_2,w_3$ each with capacity one. 

{\centering
\begin{tabular}{l@{\;}l@{\quad}|@{\quad}l}
  \multicolumn{2}{l}{Students} & Schools\\
  $u_1 \colon$ & $w_1\succ w_3 \succ w_2$ & $w_1 \colon u_2\succ u_4 \succ u_1 \succ u_5 $ \\
  $u_2\colon$ & $w_2 \succ w_1 \succ w_3 $ & $w_2\colon u_1 \succ u_2 \succ u_3 \succ u_4 \succ u_5$\\
  $u_3\colon$ & $w_2 \succ w_3$ & $w_3\colon u_3 \succ u_1 \succ u_2$\\
  $u_4, u_5\colon$ & $w_1\succ w_2$\\
\end{tabular}\par}

Here, $x\succ y$ means $x$ is preferred to $y$ by the corresponding agent. An agent not appearing in a preference list means she is not acceptable. 
The only stable matching~$\mu_1$ is to match $u_1$ with $w_2$, and $u_2$ with $w_1$, and $u_3$ with $w_3$, leaving $u_4$ and $u_5$ unmatched.
However, if $u_1$ and $u_2$ switch their partners, then they can obtain a better albeit unstable matching~$\mu_2$ (it is blocked by $\{u_4,w_1\}$).
If we increase the capacity of $w_1$ to two, then assigning $u_1$ and $u_4$ to $w_1$, $u_i$ to $w_i$, $i\in \{1,2\}$, yields a stable and efficient matching~$\mu_3$.
If we increase the capacity of $w_1$ to three, then we can extend $\mu_3$ to a stable, efficient, and perfect matching by additionally assigning $u_5$ to $w_1$.

Clearly, if we increase each school's capacity to $|U|$ so every student is assigned her first choice, then we obtain a stable, efficient, and perfect matching.
However, this  is certainly not cost effective, so we are facing the following question:
\begin{quote}
 \it How can we modify the capacities as little as possible to obtain a stable and efficient, or stable and perfect matching?
\end{quote}
In this  paper, we aim to answer this  question computationally, and look at the two most common cost functions, the total and maximum capacity 
\blu{changes} of all schools. 





\smallskip
\noindent \textbf{Our contributions.}
\blu{In the first part of the paper, we} introduce and thoroughly investigate the computational complexity of determining an optimal capacity increase vector for obtaining a stable and efficient (resp.\ stable and perfect, stable and student-popular) matching.
We consider two optimality criteria: minimizing the sum of capacity increases and the maximum capacity increase.
This gives rise to six problems: \sumse, \maxse, \sumsta, \maxsta, \sumstud, \maxstud. We show that except for \maxsta\ and \maxstud\ the other four problems are computationally hard, and remain hard even when the preference lists have constant length. 
We further search for parameterized and approximation algorithms.
For \sumsta, we prove a structural result (\cref{lemma:minsum-just-envy}) which helps constructing parametrized and approximation algorithms. 
After observing that \sumse\ and \sumsta\ can be solved in polynomial time if the capacity bound~$\sumcap$ is a constant (i.e., in XP wrt.~ $\sumcap$), we show that this  result is essentially tight since it cannot be improved to obtain an $f(\sumcap)\cdot (|U|+|W|)^{O(1)}$-time algorithm.
For the combined parameter ``the number~$|\Unas|$ of initially unmatched students'' and ``the length~$\studlen$ of the longest preference list of any unmatched students'', \sumsta\ is fixed-parameter tractable while \sumse\ remains NP-hard. For the NP-hard \sumsta, we give an integer programming formulation. 
While all three NP-hard problems are in general hard to approximate to any constant factor, \sumsta\ admits an $|\Unas|/c$-approximation algorithm for any constant $c$ (i.e. the running time is exponential in $c$), and \sumse\ and \maxse\ cannot be approximated within a factor of $(|U|+|W|+|E|)^{1-\varepsilon}$ (for any constant~$\varepsilon > 0$), where $E$ is the set of edges in the acceptability graph $G=(U,W,E)$.

Regarding stable and student-popular matchings, we show that \sumstud\ is also NP-hard and does not admit any polynomial-time constant factor approximation algorithm. On the positive side, we provide a polynomial time algorithm for \maxstud.

\blu{Finally, in the second part of the paper, we give examples to illustrate that allowing decreases in the capacities too may be beneficial for some of the problems but we obtain similar algorithmic and complexity results for this modified model.}

\begin{table}
\centering
\begin{tabular}{ |c||c|c|c|c|c|c|}
\hline 

& \textbf{const-apx} & \textbf{$|\Unas|$ -apx} & \textbf{const \studlen} & \textbf{const $|\Unas |$ } & \textbf{const $|\Unas |+\studlen$}\\
\hline
\sumsta & NP-h [T\ref{thm:sumsta-inapprox}]  & P [T\ref{thm:sumsta-unas-approx}]  & NP-c [T\ref{thm:sumsta-constdegree}] & XP [T\ref{thm:sumsta-XP-fpt-unassiged-len}], W[1]-h [T\ref{sum-stable:w1h-cap}] & FPT [T\ref{thm:sumsta-XP-fpt-unassiged-len}]\\
\hline
\maxsta & P (c=1) [T\ref{thm:maxsta}] & P [T\ref{thm:maxsta}] & P [T\ref{thm:maxsta}] & P [T\ref{thm:maxsta}]& P [T\ref{thm:maxsta}]\\
\hline
\sumse & NP-h [T\ref{thm:maxse}] & NP-h [T\ref{thm:maxse}] & NP-c [T\ref{thm:maxse}]& NP-c [T\ref{thm:maxse}]& NP-c [T\ref{thm:maxse}]\\
\hline
\maxse  & NP-h [T\ref{thm:maxse}] & NP-h [T\ref{thm:maxse}] & NP-c [T\ref{thm:maxse}]& NP-c [T\ref{thm:maxse}]& NP-c [T\ref{thm:maxse}]\\
\hline
\sumstud & NP-h [T\ref{thm:sumstud}] & NP-h [T\ref{thm:sumstud}] & NP-c [T\ref{thm:sumstud}] & NP-c [T\ref{thm:sumstud}] & NP-c [T\ref{thm:sumstud}]\\
\hline 
\maxstud & P (c=1) [T\ref{thm:maxstud}] & P [T\ref{thm:maxstud}] & P [T\ref{thm:maxstud}] & P [T\ref{thm:maxstud}] & P [T\ref{thm:maxstud}] \\
\hline
\end{tabular}

\caption{An overview of our complexity results. Here, ``const-apx'' means the complexity of an algorithm with any constant $c$ approximation factor, ''$|\Unas |$-apx'' and ''$\studlen$-apx'' is the complexity of a $|\Unas |$ or $\studlen$ approximation algorithm, finally ''const $|\Unas|$'' and ``const $\studlen$'' refer to the complexity of the decion problems when the number of initially unmatched students is bounded by a constant or   the length of the longest preference list of an unassigned student has constant value respectively.}\label{tab:results}
\end{table}

\smallskip
\noindent \textbf{Related work.}
Studying the trade-off and tension between stability and efficiency has a long tradition in Economics~\cite{SoenmezAbdul2003,ACPRT2020,EE2010eff,Ergin2002,Kesten2010,Roth1982Efficient}, but also in Computer Science~\cite{BaBa2004studentadmission,KHHSUY2017quota}.
For instance, Abdulkadiroglu and Sönmez \cite{SoenmezAbdul2003} examine Gale/Shapley's student-proposing deferred acceptance algorithm (which always yields a student-optimal stable matching) and the simple top trading cycles algorithm (which is efficient). Ergin
\cite{Ergin2002} characterizes priority structures--the so-called acyclic structure--under which a stable and efficient matching always exists. 
Kesten \cite{Kesten2010} proposes an efficiency-adjusted deferred acceptance algorithm to obtain an unstable matching which is efficient and Pareto superior to the student-optimal stable matching. 
We are not aware of any work on achieving stable and efficient matching via capacity increase. 

Chen, Skowron and Sorge \cite{ChenSkowronSorge2019teac-robust} investigate the trade-off between stability and perfectness in one-to-one matchings which is to find a perfect matching that becomes stable after a few modification to the preference lists. Our model differs from theirs as we do not allow modifications to the agents' preferences. Limaye and Nasre
\cite{LN2022} introduce two related matching problems, where each school has unbounded capacity and a value that measures the cost of assigning an (arbitrary) student there,
and the goal is to find a stable and perfect matching with minimum sum of costs or minimum maximum cost.
Their models are different from ours since they assume that any place at a given school must have the same cost, whereas in our framework  each school has a capacity~$q$ so that the first $q$ places are considered free, and only the additional places have non-zero cost. Furthermore, they allow no initial quotas, which may be interpreted in our framework as setting the initial quota of each school to be zero. 

Recently, capacity variation in many-to-one matching has been studied, albeit for different objectives. Ríos et al. 
\cite{RLPC2014} propose a seat-extension mechanism to increase student's welfare. Ueda et al. 
\cite{KHHSUY2017quota} design a strategy-proof mechanism to address minimum and maximum quotas.
Nguyen and Vohra \cite{NV2018} study many-to-one matching with couples and propose algorithms to find a stable matching by perturbing the capacities.
Bobbio et al. \cite{BCLT2022capvariation,BCLRT2022capplanning} consider capacity variations to obtain a stable matching with minimum sum of the ranks of the matched schools (\avgrank) or maximum cardinality (\maxsm). The capacity variations can be either sum of capacity increases or sum of capacity decreases. Our problem \sumsta\ can be reduced to their \minsum\avgrank\ problem by introducing sufficiently many dummy students and garbage collector schools with very high ranks so to ensure that each original student is matched. Hence, our hardness results apply to that problem. They left open the complexity of \minsum\maxsm, which is NP-hard by our hardness for \sumsta. Abe, Komiyama and Iwasaki
\cite{AKI2022CapExp} propose some alternative method and conduct experiments for \minsum \avgrank. Yahiro and Yokoo \cite{yahiro2020game} considered the problem where there are a fixed number of resources and the capacities of the schools are based on how much resource we allocate them. In their framework the resources can only be allocated to specific schools, and these relations were given by a bipartite graph. We can reduce our problem of minimizing the sum of capacity increases needed by having resources for each school that can only be allocated to that school and having $\sumcap$ other resources that can be allocated to any school. In their paper, Yahiro and Yokoo studied different mechanisms and analyzed whether they are fair (stable), Pareto-efficient or strategy proof and also gave some impossibility results. None of the results of this  paper follows from their results.

\smallskip
\noindent \textbf{Paper structure.}
After \cref{sec:defi}, where we define relevant concepts and notation,
we consider stable and perfect matchings in \cref{sec:SP}, stable and efficient matchings in \cref{sec:SE}, and stable and student-popular matchings in \cref{sec:STUPOP}. \blu{Then we extend our results to the case when decreases are also allowed in \cref{sec:decr}.}
We conclude with some future research directions in \cref{sec:conclude}.

\section{Basic definitions and fundamentals} 
\label{sec:defi}

By $\mathds{N}$ we mean the set of all positive integers. 
Given an integer~$t$, let \myemph{$[t]$} (without any prefix) denote the set~$\{1,\cdots,t\}$.
Given two integer vectors~$\px, \py$ of dimension~$t$, i.e., $\px,\py \in \mathds{Z}^{t}$,
we let \myemph{$\px + \py$} denote the new integer vector~$\pr$ with $\pr [z] = \px[z]+\py[z]$ for all $z\in [t]$,
and we write \myemph{$\px\le \py$} if for each index~$i\in [z]$ it holds that
$\px[i]\le \py[i]$; otherwise, we write \myemph{$\px \not\le \py$.}
A \myemph{preference list} (or \myemph{priority order})~$\succ$ over a set~$A$ is a linear order over~$A$.
We say that~$x$ is \myemph{preferred} to~$y$ if $x\succ y$.
Given a finite set~$A$, we write~\myemph{$\seqq{A}$} to denote an arbitrary but fixed linear order on~$A$.
If $A=\{a_1,\ldots,a_{|A|}\}$,  then $\seqq{A}$ always denotes the increasing order~$a_1,\ldots,a_{|A|}$.

\subsection{Many-to-one matching}
The \manyone\ (in short, \manyones) problem has as input a set $U=\{u_1,\ldots, u_{n}\}$ of $n$ students and a set~$W=\{w_1,\ldots,w_{m}\}$ of $m$ schools, together with the following information.
\begin{itemize}[--]
  \item Each student~$u_i\in U$ has a \myemph{preference list $\succ_{u_i}$} over a subset of schools; the agents in $\succ_{u_i}$ are called \myemph{acceptable} to~$u_i$.
  \item Each school~$w_j\in W$ has a \myemph{priority order $\succ_{w_j}$} over a subset of students.
  \item A \myemph{capacity vector~$\quota\in \mathds{N}^{m}$} which specifies
  the maximum number of students allowed to be admitted to each school~$w_j \in W$. 
\end{itemize}
Note that following the literature~\cite{GusfieldIrving1989,Ergin2002,Manlove2013}, we assume each school to have capacity at least one.
For each student $x\in U$ (resp.\ school~$x\in W$), let \myemph{$\acset(x)$} denote the \myemph{acceptable set} of~$x$, which contains all schools acceptable to student~$x$ (resp.\ all students contained in the priority order of school~$x$).
Throughout, we assume that no student (resp.\ no school) has an empty preference list (resp.\ priority order), and
for each student~$u$ and each school~$w$ it holds that $u$ finds $w$ acceptable if and only if $u$ is also in the priority order of~$w$. We call a student $u$ a \myemph{admirer} of school $w$, if $w$ is the best school for $u$ according to $\succ_u$.
We can also model the acceptability relations with a bipartite graph $G=(U,W;E)$, where $(u_i,w_j)\in E$ if and only if $u_i$ and $w_j$ find each other acceptable. 

A \myemph{matching}~$\mu$ is a (partial) mapping from $U$ to $W$ such that each student~$u\in U$ is either unmatched (i.e., $\mu(u)$ is undefined) or matched to an acceptable school (i.e., $\mu(u)\in \acset(u)$), or equivalently a subset of the edges of the acceptability graph where each student $u$ has at most one incident edge. If $u$ is unmatched, we also write \myemph{$\mu(u)=\bot$}.
If $\mu(u)=w$, then we say that $u$ and $w$ are \myemph{assigned} (or \myemph{matched}) to each other. 
We say that school~$w_j$ is \myemph{under-filled} if $|\mu^{-1}(w_j)| < \quota[j]$; otherwise $w_j$ is \myemph{full} .
We assume that each student~$u$ prefers to be assigned rather than unassigned,
i.e., $w\succ_u \bot$ for all $w\in \acset(u)$, and each school~$w$ prefers to admit more acceptable student rather than remain under-filled.

Given a  \manyones\ instance~$I=(U, W, (\succeq_x)_{x\in U\cup W}, \quota)$ and a matching~$\mu$ for $I$,
we say $\mu$ is \myemph{perfect} if every student is matched under~$\mu$,
and that it is \myemph{feasible} if each school $w_j\in W$ is assigned at most $\quota[j]$ students, i.e., $|\mu^{-1}(w_j)| \le \quota[j]$.

\smallskip
\noindent \textbf{Justified envy, blocking pairs, and stable matchings.}
Let $\mu$ be a matching.
A student~$u$ is said to have \myemph{justified envy} towards (or \myemph{justified envies}) another student~$u'$ under~$\mu$ if it holds that $w=\mu (u')\succ_u\mu (u)$ and $u\succ_w u'$.
In this  case, student~$u$ (resp.\ matching~$\mu$) is said to have a justified envy.
A student~$u$ and a school~$w$ form a \myemph{blocking pair} in~$\mu$ if the following conditions are all satisfied: 
\begin{compactenum}[(i)]
  \item $u \in \acset(w)$, 
  \item $u$ is either unmatched or prefers $w$ to~$\mu (u)$, and
  \item $w$ is under-filled or prefers $u$ to at least one of its assigned students. 
\end{compactenum}
Consequently, matching~$\mu$ is \myemph{stable} if no student forms a blocking pair with any school.
Clearly, in a stable matching, there is no justified envy.
Matching~$\mu$ is called a \myemph{student-optimal stable} matching
if 
it is stable and every student \myemph{weakly prefers} $\mu$ to all other stable matchings, i.e., for all students~$u$ and all other stable matching~$\sigma$ it holds that either $\mu(u)=\sigma(u)$ or $\mu(u)\succ_u \sigma(u)$. It is well known that such a student-optimal stable matching always exists~\cite{GusfieldIrving1989}.

\smallskip
\noindent \textbf{Dominance and efficient matchings.}
A matching~$\mu$ \myemph{dominates} another matching~$\sigma$ if the following holds:
\begin{compactitem}[--]
  \item for each student~$u$ it holds that either $\mu(u)=\sigma(u)$ or $\mu(u)\succ_u \sigma(u)$, and
  \item at least one student~$u$ has $\mu(u)\succ_u \sigma(u)$.
\end{compactitem}
Clearly, a student-optimal stable matching is a stable matching which is not dominated by any other stable matching, but it can be dominated by other matchings.
We call a matching~$\mu$ \myemph{efficient} (aka.\ \myemph{Pareto-efficient}) if it is not dominated by other matchings.
It is well-known that stability and efficiency are not compatible with each other, even in the restricted one-to-one matching case, i.e., each school has capacity one.
\paragraph{Student-popular matchings.}

Next we define student-popularity. For this, we compare two matchings $\mu $ and $\mu '$ in the following way. Each student $u\in U$ cast a vote $vote_u(\mu ,\mu ')\in \{ -1,0,1\}$, such that $vote_u(\mu ,\mu ')=0$, if $\mu (u)=\mu ' (u)$, $vote_u(\mu ,\mu ')=1$, if $\mu (u)\succ_u\mu ' (u)$ and $vote_u(\mu ,\mu ')=-1$ otherwise. Then, we say that a matching $\mu $ is \emph{student-popular}, if there is no other matching $\mu '$, which strictly beats $\mu$ in this voting. That is, $\mu$ is student-popular if $\sum_{u\in U}vote_u(\mu ,\mu ')\ge 0$ for any matching $\mu '$.

\subsection{Optimal matchings via capacity increase}
We consider the following decision problems, where $\Pi \in \{$stable and perfect, stable and efficient, 
stable and student-popular
$\}$.

\decprob{\minsummatching ({\normalfont\text{resp. }}\minmaxmatching)}
{%
  An instance~$I=\instance$ of \manyones, a capacity bound~$\sumcap\in \mathds{N}$ (resp.\ $\maxcap\in \mathds{N}$).%
}
{%
  Is there a capacity increase vector~$\pr$ with $\lone\le \sumcap$ (resp.\ $\lmax \le \maxcap$) s.t.\ $I'=(\profile, \quota+\pr)$ admits a $\Pi$~matching. 
}


We abbreviate the problems \minsummatching[Stable and Perfect], \minmaxmatching[Stable and Perfect], \minsummatching[Stable and Efficient], \minmaxmatching[Stable and Efficient],
\minsummatching[Stable and Student-popular],
\minmaxmatching[Stable and Student-popular]
with \sumsta,  \maxsta, \sumse, \maxse, 
\sumstud, and \maxstud\ respectively.

For brevity's sake, we also use the same names to refer to the corresponding minimization variant.
For each of the above problems,
we call a capacity increase vector~$\pr$
\begin{compactitem}[--]
  \item \myemph{feasible} if it results in a solution for the corresponding problem,
  \item \myemph{good} if it is a witness for a YES answer, and
  \item \myemph{optimal} if $\lone$ (resp.\ $\lmax$) is minimum among all feasible vectors for \minsum\ (resp.\ \minmax).
\end{compactitem}

\textcolor{blue}{In Section \ref{sec:decr} we consider the variants of these problems, where decreases are also allowed in the capacities, which corresponds to allowing $\pr$ to be negative, so in that case, we call $\pr$ a \emph{capacity change} vector instead of a capacity increase vector. }

It is straightforward that \sumsta\ and \maxsta\ are contained in NP.
NP-containment of the other problems are based on the following.

\begin{lemma}
\label{lemma:sumse-inNP}
    Checking whether an \manyones\ instance admits a stable and efficient or stable and student-popular matching can be done in polynomial time.
\end{lemma}
\begin{proof}
  First, observe that a stable and efficient or stable and student-popular matching must be the student-optimal stable one as otherwise the student-optimal stable matching Pareto-dominates it (hence also beats it in the voting).
  So it is enough to show that we can check whether the student-optimal stable matching is Pareto-efficient or student-popular for the students since the student-optimal stable matching can be found in linear time.
  For both of these, we can treat the instance as an instance with one-sided preferences.
  By the result of Cechlárová et al. \cite{cechlarova2014pareto}, checking Pareto-efficiency for one-sided preferences can be done in polynomial time. By Biró, Irving and Manlove \cite{biro2010popular}, checking student-popularity can also be done in polynomial-time. 
\end{proof}

First of all we state a trivial, but important observation. 
\begin{observation}
Any student-popular matching $\mu$ is also efficient.
\end{observation}
Hence, the problem of finding a student-popular and efficient matching is just the same as finding a student-popular one.

The example in the introduction discusses stable and efficient matchings. Another example is the following.
\begin{example}\label{ex:problems}
  Suppose we have three schools $w_1,w_2,w_3$, each with capacity one, and five students $u_1,u_2,u_3,u_4,u_5$.
  Let the priority lists and preference lists be as follows, where $i\in [5]\setminus \{3\}$ and $j\in [2]$.
  
  \begin{tabular}{l@{\quad}|@{\quad}l}
      Students & Schools\\
      $u_{i} \colon w_1\succ w_2$ &  $w_j \colon u_1\succ u_2 \succ u_3 \succ u_4 \succ u_5$ \\
      $u_3\colon w_1 \succ w_2 \succ w_3$ & $w_3\colon u_3$\\
    \end{tabular}

    \noindent The student-optimal (and only) stable matching is $\mu =\{ (u_1,w_1)$, $(u_2,w_2),(u_3,w_3)\}$.
  Students~$u_4$ and $u_5$ are the two unassigned students.
  If we assign both $u_4$ and $u_5$ to $w_1$, 
  then we also need to reassign $u_2$ and $u_3$ to $w_1$ since otherwise both $u_2$ and $u_3$ would have justified envies towards $u_4$ and $u_5$.
  This  amounts to a capacity increase of four.
  If we assign both $u_4$ and $u_5$ to $w_2$ however, then only student $u_3$ would have justified envies.
  Now, we can reassign $u_3$ to $w_2$ as well (increasing the capacity of $w_2$ by three) and obtain a stable matching which matches every student.
  One can verify that no smaller capacity increase than three is possible to obtain a stable and perfect matching.
\end{example}

\smallskip
\noindent \textbf{Notation.}
By the well-known Rural Hospitals theorem~\cite{GusfieldIrving1989}, every \manyones-instance admits a stable matching and all stable matchings match the same set of students and schools. Hence, if not stated otherwise, given an \manyones-instance, let \myemph{$\As$} and \myemph{$\Unas$} denote the set of assigned and unassigned students in a stable matching, respectively. Furthermore, let \myemph{$\studlen$}
denote the maximum length of the preferences of all unassigned students in~$\Unas$, respectively.
Finally, let \myemph{$\scholen$} denote the maximum length of the priority lists of the schools.

\subsection{Structural properties}
Kominers \cite{Kominers2020} shows that increasing the capacity of one school by one weakly improve the students' outcome.
Due to her, we can observe the following.
\begin{lemma}\label{lem:capincrease}
  Let $I_1=(U,W,(\succ_x)_{x\in U\cup W}, \pq_1)$ and  $I_2=(U,W,(\succ_x)_{x\in U\cup W}, \pq_2)$ denote two \manyones-instances with the same sets of students and schools, and the same preferences and priority lists such that $\pq_1\le \pq_2$.
  Then, the following holds.
  \begin{compactenum}[(i)]
    \item\label{lem:improve}  Every student weakly prefers the student-optimal stable matching~$\mu_2$ in~$I_2$ to the student-optimal stable matching~$\mu_1$ in~$I_1$.
    \item\label{lem:underfilled} If a school~$w$ is under-filled in $\mu_1$,
    then $\mu_2^{-1}(w)\subseteq \mu_1^{-1}(w)$ holds.
  \end{compactenum}
\end{lemma}

\begin{proof}
  Let $I_1$ and $I_2$ be as defined.
  Let $U$, $W$, $\pq_1$, $\pq_2$ be the sets of students and schools, the capacity vectors for $I_1$ and $I_2$, respectively.
  To show the first statement, it suffices to consider the case when exactly one school has increased its capacity and this  increase is one since the weakly preferring relation is transitive, which has been shown by Kominers \cite{Kominers2020}.

  For Statement~\eqref{lem:underfilled}, suppose, for the sake of contradiction, that there exists a student~$u\in \mu_2(w)\setminus \mu_1(w)$.
  Note that by the first statement, $u$ strictly prefers $\mu_2(u)=w$ to $\mu_1(u)$.
  Since $w$ is under-filled in~$\mu_1$, it forms with $u$ a blocking pair in~$\mu_1$, a contradiction.
\end{proof}

We proceed to a crucial observation that helps us decrease the number of edges in an input instance of \sumsta\ or \maxsta.

\begin{observation}
\label{obs:deleteedges}
    The optimum values of \sumsta\ and \maxsta\ do not change, if we delete each edge $(u,w)$ from the instance, where $\stuopt (u) \succ_u w$ for the student-optimal stable matching $\stuopt$.
\end{observation}
Indeed, for any capacity increase vector $\pr$, we have by Lemma \ref{lem:capincrease} that each student $u$ will be admitted to a school at least as good as $\stuopt [u]$ if the capacities are increased with $\pr \ge 0$, so the deletion of these edges do not decrease the size of the stable matching. This leads to the following crucial property we can assume of any input instance of \sumsta\ or \maxsta.

\begin{corollary}
\label{cor:uniquestable}
    For \sumsta\ and \maxsta\, we may assume that in the input instance $I$, there is a unique (so both school and student-optimal) stable matching, where each student is unmatched or is at her worst acceptable school.  
\end{corollary}

Next, we introduce a some new concepts and prove a useful lemma for designing further parameterized and approximation algorithms. 
Let $\stuopt$ be the student-optimal stable matching of a given instance $I$. 
Consider all possible vectors~$\pv$, whose coordinates are the unassigned students and for each coordinate~$u$, the entry is an acceptable school for the unassigned student~$u$.
Denote the set of such vectors by \myemph{$\Vecs$}.

\newcommand{\n}{n}
\newcommand{\worst}{wo}
Let $\pv \in \Vecs$ be such a vector. We create a new instance $I(\pv)$ as follows. For each school $w$, let $\n(w,\pv )$ be the number of coordinates of $\pv$ that is equal to $w$. Decrease every school $w$'s capacity to $(\quota [w]-\n (w,\pv ))^+ =\max \{ 0, \quota [w]-\n(w,\pv )\}$. Delete every $(u,w)$ edge, where $u\in \As$ and there is an $u'\in \Unas$ such that $w\succ_{u'} \pv [u']$ and $u'\succ_w u$. Then, delete every edge $(u,w)$ with $u\in \As$ such that there is a student $u'$ with $\pv [u']\succ_u w$ and $u\succ_{\pv [u']}u'$. We also delete each student in $\Unas$.
 If we delete all edges incident to a student $u$, then we do not define $I(\pv )$. Hence, suppose that all $u\in \As$ still have acceptable schools.
This concludes the creation of the new instance $I(\pv)$.
Then, for each student $u\in \As$, let $\worst (u, \pv)$ be the worst remaining school $w$ for $u$ in $I(\pv )$. 

\newcommand{\mmid}{\hat{\mu}_{\pv}}
Finally, let $\JE (\pv)$ denote the set of students that remain unmatched in the student-optimal stable matching $\mmid $ of $I(\pv)$ and let $\je (w,\pv )$ denote $|\{ u\in \JE( \pv) \mid  \worst (u,\pv ) =w\} |$. 

\newcommand{\mnew}{\mu^*_{\pv}}
Create a matching $\mnew$ as follows. Match every $u\in \Unas$ to the school corresponding to her coordinate in $\pv$. Match every $u\in \As \setminus \JE (\pv)$ to the school she is matched to in $\mmid $. Finally, match each $u\in \JE (\pv) $ to $\worst (u,\pv)$ and increase the capacities of the schools if needed, to the number of students admitted. Denote the new instance created by $I^*(\pv)$. It is easy to see that each school $w$'s capacity becomes $\max \{ \n (w,\pv) +|\mmid (w)| + \je (w,\pv), \quota [w]\}$. 

\newcommand{\Vecstable}{\mathcal{V}'}

For clarity, we provide a pseudocode on how to create the instance $I^*(\pv )$ and the matching $\mnew $ in Algorithm \ref{alg:procedure}.

\begin{algorithmic}[h]

\begin{algorithm}
\caption{Create $I^*(\pv)$ and $\mnew $} \label{alg:procedure}   

\STATE Input: $\pv \in \Vecs$, an instance $I$ of \manyone\ and the sets $\As,\Unas$
Output: $I^*(\pv),\mnew$
\STATE Set $E=E':=\{ (u,w) \mid u\in \acset (w) \wedge w\in \acset (u)\}$

\FOR{Each school $w\in W$}
\STATE $n(w,\pv)$:= number of coordinates of $\pv$ that equals $w$
\STATE Set $\quota '[w]=\max \{ \quota [w]-n(w,\pv ),0\}$
\ENDFOR
\FOR{$(u,w)\in E'$ with $u\in \As$}
\IF{there is an $u'\in \Unas$ s.t. $w\succ_{u'}\pv [u'] \wedge u'\succ_w u$ }
\STATE Delete $(u,w)$ from $E'$
\ELSIF{there is an $u'\in \Unas$ s.t. $\pv [u']\succ_u w \wedge u\succ_{\pv [u']} u'$ }
\STATE Delete $(u,w)$ from $E'$
\ENDIF
\ENDFOR
\IF{There is a $u\in \As$ with no remaining incident edges}
\STATE \textbf{Return} "No $I^*(\pv )$ instance"
\ENDIF
\FOR{$x\in \As \cup W$}
\STATE Set $\succ'_x$ to be the restriction of $\succ_x$ to their remaining partners in $\As \cup W$ 
\ENDFOR
\STATE Set $I(\pv ) = (\As , W, (\succ_x)_{x\in \As \cup W}, \quota') $
\STATE $\mmid  :=$ student-optimal stable matching of $I(\pv )$
\STATE $\JE (\pv ):=$ unmatched students in $\mmid $ and $\je (w,\pv):=|\{ u\in \JE( \pv) \mid  \worst (u,\pv ) =w\} |$
\FOR{$u\in \As$}
\STATE $\worst (u,\pv) :=$ the worst remaining school of $u$ in $I(\pv)$
\ENDFOR
\FOR{$u\in U$}
\IF{$u\in \Unas$}
\STATE $\mnew  (u) := \pv [u]$
\ELSIF{$u\in \JE (\pv )$}
\STATE $\mnew  (u):= \worst (u,\pv)$
\ELSE
\STATE $\mnew (u):=\mmid (u) $

\ENDIF
\ENDFOR
\FOR{$w\in W$}
\STATE $\quota^* [w]:=\max \{ \n (w,\pv) +|\mmid (w)| + \je (w,\pv), \quota [w]\}$
\ENDFOR
\STATE \textbf{Return} $I^*(\pv )=(U,W,(\succ_x)_{x\in U\cup W},\quota^*)$ and $\mnew $
\end{algorithm}
\end{algorithmic}

\begin{lemma}
\label{lemma:minsum-just-envy}
Let $\Vecstable \subseteq \Vecs$ be the set of those $\pv$ vectors, such that $I^*(\pv)$ exists and $\mnew $ is stable in $I^*(\pv)$.
Then, the sum of values of an optimal capacity increase vector for \sumsta\  is equal to $\min\limits_{\pv\in \Vecstable}\{ \sum_{w\in W}(\n (w,\pv) +|\mmid (w)| + \je (w,\pv) - \quota [w])^+\}$.
\end{lemma}
\begin{proof}
 Let $I=(U,W,(\succ_x)_{x\in U\cup W}, \pq)$ be an instance of \manyones, and let $\pr$ be a capacity increase vector with minimum $\lone$ such that $(U,W,(\succ_x)_{x\in U\cup W}, \pq+\pr)$ admits a stable and perfect matching. By Corollary~\ref{cor:uniquestable}, we assume that in $I$, there is a unique (so also student-optimal) stable matching $\stuopt$, where every $u\in U$ is either unmatched or matched to her worst acceptable school.

First, we show that $\min\limits_{\pv\in \Vecstable}\{\sum_{w\in W}(\n (w,\pv)+|\mmid (w)| +\je (w,\pv)- \quota [w])^+\}\ge \lone$. 
Let $\pv \in \Vecstable$. Then, by the above, we can create a new instance $I^*(\pv)$, which admits a stable matching $\mnew $, where every student is assigned. As school $w$'s capacity becomes $\max \{ \n (w,\pv) +|\mmid (w)| + \je (w,\pv), \quota [w]\}$, we have that we need to increase $\quota [w]$ by $(\n (w,\pv) +|\mmid (w)| + \je (w,\pv) - \quota [w])^+$. So, there is a feasible capacity increase vector, where the total increase is at most $ \sum_{w\in W}(\n (w,\pv) +|\mmid (w)| + \je (w,\pv) - \quota [w])^+$. 

\newcommand{\stuoptR}{\hat{\mu}_{\quota +\pr}}
Next, we show that $\min\limits_{\pv\in \Vecstable}\{ \sum_{w\in W}(\n (w,\pv) +|\mmid (w)| + \je (w,\pv) - \quota [w])^+\}\le \lone$.
Consider the optimal capacity increase vector $\pr$ and let $\stuoptR$ be the student-optimal stable matching in $(U,W,(\succ_x)_{x\in U\cup W}, \pq+\pr)$. Take the vector $\pv \in \Vecs$, where each $u\in \Unas$ is assigned to the same school as in $\stuoptR$. Clearly, then it is true that there is no justified envy among the students in $\Unas$ if they are assigned according to $\pv$. It is also clear that in $\stuoptR$, (i) each student $u\in \As$ is at a school $w$ that no student $u'\in \Unas$ with $u'\succ_w u$ envies and (ii) each student $u\in \As$ is in a school that is at least as good as $\worst (u,\pv )$ by the stability of $\stuoptR$. This implies that if we restrict the instance to $\As \cup W$, delete each edge $(u,w)$, if there exists a student $u'\in \Unas$ such that $u'\succ_w u$, $w\succ_{u'}\pv [u']$ or there is a $u'\in \Unas$ such that $\pv [u']\succ_u w $ and $ u\succ_{\pv [u']} u'$ ( meaning that $\worst (u,\pv) \succ_u w$) and set each school $w$'s capacity to $\quota [w] +\pr [w] - \n (w,\pv ) \ge 0$, then there is a perfect stable matching in that instance.
 Denote this instance $\hat{I}$.

In the instance $I(\pv)$, the capacity of each school $w$ is exactly $ (\quota [w]-\n (w,\pv))^+$, which is an upper bound on how many original seats are still available at $w$ without capacity increase, after assigning the students in $\Unas$ according to $\pv$. As $(\quota [w] -\n (w,\pv))^+ \le \quota [w] +\pr [w] -\n (w,\pv )$ for each $w$, and in the instance $I(\pv)$, we have the same set of schools, students and edges as in $\hat{I}$ and there are $\sum_{w\in W} \je (w,\pv)$ students that are unassigned in the student-optimal stable matching of $I(\pv)$, we need at least $\sum_{w\in W} \je (w,\pv)$ new places in $I(\pv)$ to have a stable matching where every student (in $\As$) is assigned, and hence

$$\sum_{w\in W}\quota [w] +\pr [w] - \n (w,\pv) - (\quota [w]-\n(w,\pv ))^+ \ge \sum_{w\in W} \je (w,\pv). $$

Rearranging this, we get that

$$\sum_{w\in W}\pr [w] \ge \sum_{w\in W} \je (w,\pv) -\quota [w] + \n (w,\pv) +(\quota [w]-\n (w,v))^+$$.

Using that $|\mmid (w) |\le (\quota [w] -\n (w,\pv) )^+$ and $\je (w,\pv) -\quota [w] + \n (w,\pv) +(\quota [w]-\n (w,v))^+ \ge 0$, as $\je (w,\pv) \ge 0$, we get that 

$$\sum_{w\in W}\pr [w] \ge \sum_{w\in W}(\n (w,\pv) +|\mmid (w)| + \je (w,\pv) - \quota [w])^+.$$

As $\lone = \sum_{w\in W}\pr [w]$, it only remains to show that $\pv \in \Vecstable$.

By our above observations about $\stuoptR$, we get that in the instance $\hat{I}$, the restriction of $\stuoptR$ is a stable matching that assigns every student from $\As$ and fills every school $w$ that is envied by a student $u\in \Unas$. To see that the restriction of $\stuoptR$ to $\hat{I}$ is stable, note that the same set of schools are underfilled, and hence a blocking edge for this matching must also be blocking for $\stuoptR$. In particular, this implies that $I(\pv ) $ and hence $\mnew $ exist. Hence, we only have to show the stability of $\mnew $.

In the instance $I(\pv)$, we have that each school's capacity is weakly smaller than $\quota [w]+\pr [w]- \n (w,\pv)$, which is $w$'s capacity in $\hat{I}$. Also, as we have shown, the student-optimal stable matching in $\hat{I}$ fills every school that is envied by some $u\in \Unas$ in $\stuoptR$. Combining these we get that the student-optimal stable matching $\mmid $ in $I(\pv)$ also satisfies that it fills 
every school that is envied by some $u\in \Unas$ in $\stuoptR$ (if we decrease a school's capacity, then if it was full in the student-optimal stable matching then it remains full by \cref{lem:capincrease} (\ref{lem:underfilled})). 

Therefore, in $\mnew $ we have that for any student $u\in \Unas$ and school $w\succ_u \stuoptR (u)$, $w$ is full in $\mnew $. Also, by the definition of the edges in $I(\pv)$ and the schools $\worst (u', \pv)$ we get that all students from $\As$ at $w$ in $\mnew $ are preferred to $u$ by $w$. Hence, together with the fact that there is no justified envy between students from $\Unas$, we conclude that there can be no blocking edge $(u,w)$ to $\mnew $, where $u\in \Unas$.

Similarly, if there is an edge $(u,w)$ with $u\in \As$ that got deleted in $I(\pv)$, then either $u$ got admitted to a school that she prefers to $w$, if $\worst (u,\pv)\succ_u w$ or there is a student $u'\in \Unas$ such that $u'$ prefers $w$ to $\stuoptR (u')=\mnew (u')$ and $u'\succ_w u$. Therefore, if such an edge $(u,w)$ blocks $\mnew $, then there is an edge $(u',w)$ with $u'\in \Unas$ that blocks it, contradiction.

Suppose that there is a blocking edge $(u,w)$ with $u\in \As$ that also exits in $I(v)$. If $w$ is underfilled in $\mnew$, then $\quota [w] > \n (w,\pv) + |\mmid  (w)|+\je (w, \pv)$, so $w$ was underfilled in $I(\pv)$ with the matching $\mmid $ too. If $u\in \JE (\pv )$, then $(u,w)$ would block $\mmid $, contradiction. If $u\notin \JE ( \pv)$, then $u$ is assigned to the same school in $\mnew $ and $\mmid$, so $(u,w)$ would block $\mmid $ again, contradiction. Hence, $w$ is full. By the definition of the edges in $I(\pv)$, there is no student $u'\in \Unas$ at $w$ such that $u\succ_w u'$ (otherwise $w\succ_u  \mnew (u)\succeq \worst (u,\pv) \succeq w$, contradiction). So there must be a student $u'\in \As$ such that $u'\in \mnew (w)$ and $u\succ_w u'$. If $u'\in \mmid (w)$, we get that $\mmid $ is not stable, contradiction. Otherwise, $u'$ is unassigned in $\mmid$ and $w=\worst (u',\pv)$. 
Hence, $w$ is the worst entry of $u'$. By Corollary~\ref{cor:uniquestable}, we can assume that $w$ is at least as good for $u'$ as the school that $u'$ gets assigned to in the student-optimal stable matching $\stuopt$ in $I$. 
As $(u,w)$ does not block this stable matching $\stuopt$ in $I$ either, we have that $w$ is worse that the school $w'=\stuopt (u)$ of $u$. However, $u$ is admitted to a school at least as good $w'$ (edges to worse schools are deleted even from $I$), so $(u,w)$ does not block, contradiction.

We conclude that $\mnew $ is stable, and for this $\pv$, it holds that $\pv \in \Vecstable$ and $\sum_{w\in W}(\n (w,\pv ) +|\mmid (w)|+\je (w,\pv)-\quota [w])^+ =\lone$, concluding the proof.

\end{proof}

\section{Stable and Perfect Matchings}\label{sec:SP}

In this  section, we focus on stable and perfect matchings and the computation complexity regarding \sumsta\ and \maxsta. 

\subsection{Minisum capacities}\label{sec:SP+sum}

\subsubsection{Hardness results}

We start with a dichotomy result regarding the length of the priority/preference list. 
\newcommand{\vcforward}{%
   If $G$ admits a vertex cover~$V'$, then there is a good capacity increase vector~$\pr$ with $\lone \le \emm+|V'|$.
 }
 \newcommand{\vcbackward}{%
    If there exists a feasible capacity increase vector~$\pr$, then there exist a vertex cover of size at most~$\lone-\emm$.
 }
\begin{theorem}
\label{thm:sumsta-constdegree}
\sumsta\ is NP-complete; hardness remains even if $\pq=1^m$, all students have at most four acceptable schools, $\studlen = 2$ and $\scholen = 3$.
If $\studlen \le 1$ or $\scholen \le 2$, then \sumsta\ becomes polynomial-time solvable.

\end{theorem}
\begin{proof}
  We first show the polynomial results.
  Clearly, if $\studlen\le 1$, then there is only one possible assignment vector for the unassigned students, and by \cref{lemma:minsum-just-envy} the problem can be solved in polynomial time.
  Next, assume that $\scholen \le 2$.
  We aim to show that every assignment vector is ``good''. 
  Since each school has at least one seat, it must receive at least one student in each initial stable matching as otherwise by \cref{lem:capincrease}\eqref{lem:underfilled} we can ignore such school. 
  If an unassigned student would be assigned to a school, then together with the initially assigned student, there are already two, no other assigned or unassigned student would have a justified envy.
  Hence, each assignment vector directly corresponds to a good capacity increase vector.
  In other words, we only need to check whether $\sumcap \ge |\Unas|$.
  
  Now, we turn to the hardness and reduce from the NP-complete \VC\ problem~\cite{GJ79}.
  \decprob{\VC}
  {Graph~$G=(V,E)$ and a non-negative integer~$h$.}
  {Is there a \myemph{vertex cover} of size at most~$h$, i.e., a subset~$V'\subseteq V$ of size at most $h$ such that $G[V\setminus V']$ is edgeless?}
  Note that we can assume that the problem remains NP-hard even if each vertex has degree at most three~\cite{GJS1976Simplified}. 
  Let $I=(G,h)$ be an instance of \VC, where $G=(V,E)$ with $V=\{v_1,\ldots,v_{\enn}\}$ and $E=\{e_1,\ldots,e_{\emm}\}$. 
  We create an instance of our problem as follows.
  For each edge~$e_{t}\in E$, create an edge student~$e_{t}$.
  For each vertex~$v_i\in V$ and each incident edge~$e_{t}$ (i.e., $v_i\in e_{t}$),
  create a school~$v^{t}_i$ with capacity one, along with a dummy student~$d_i^{t}$. 
  For each $i\in [\enn]$, create a school~$w_i$ with capacity one, and a student~$u_i$.
  The preferences and priority lists are as follows: 

  {\centering
    \begin{tabular}{l@{\quad}|@{\quad}l}
      Students & Schools\\
      $e_{t} \colon v^{t}_i\succ v_{j}^{t}$ &  $v_i^{t} \colon d_i^{t}\succ u_i\succ e_{t}$ \\
      $u_i\colon \seqq{\{v_i^{z}\mid v_i \in e_z\}} \succ w_i$ & $w_i \colon u_i$ \\
    $d_i^{t} \colon v_i^{t}$
    \end{tabular}\par}
  
  To complete the construction, let the capacity bound be $\sumcap=\emm+h$.
  It is straightforward to verify that in the constructed instance 
  each school has capacity one and at most three students in its priority list,
  and every student has at most four acceptable schools since each vertex is incident to at most three edges.
  As for the unassigned students, we observe that the student-optimal stable matching assign each~$u_i$ to the corresponding school~$w_i$ and each dummy student~$d_i^{t}$ to the corresponding school~$v_i^{t}$, where $v_i\in V$ and $e_{t}\in E$ such that $v_i\in e_t$.
  This  means that the edge students~$e_t$ are the only unassigned students, each with two acceptable schools.
  
  \begin{claim}\label{claim:sumSTA-bounddeg-VC->good}
    \vcforward
  \end{claim} 
  
    \begin{proof}[Proof of claim~\ref{claim:sumSTA-bounddeg-VC->good}]
    \renewcommand{\qedsymbol}{$\diamond$}
    Let $V'$ be a vertex cover.
    We first define a matching~$\mu$ then we will define the capacity increase vector.
    For each vertex~$v_i\in V$ and all incident edges~$e_t$, let $\mu(d_i^t)=v_i^t$.
    For each vertex~$v_i\in V\setminus V'$, let $\mu(v_i)=w_i$.
    For each vertex~$v_i\in V'$,
    let $\mu(u_i)=v_i^{t}$, where $v_i^t$ is the
    most preferred school of $u_i$,
    and let $\mu(e_t)=v_i^t$ if the other endpoint of $e_t$ is $v_j$ such that either $i<j$ or $v_j\notin V'$ holds.
    Now, we define $\pr$.
    For each~$i\in [\enn]$, let $\pr[w_i]=0$.
    For each~$v_i\in V$ and each incident edge~$e_t\in E$ with $v_i\in e_t$,
    let $\pr[v_i^t]=|\mu^{-1}(v_i^t)|-1$.
    Clearly, $|\pr|_1\le \emm+|V'|$ since exactly $|V'|$ students~$u_i$ get rematched to their most preferred school.
    Further, $\mu$ matches every student since $V'$ is a vertex cover. 
    We claim that matching~$\mu$ is indeed stable.
    Suppose, towards a contradiction, that there is a student~$x$ and a school~$y$ that form a blocking pair.
    Clearly, $x$ cannot be a dummy student since they already obtain their respective most preferred school. If $x$ would be an edge student, than her only better school must be a $v_i^t$ school such that $v_i\notin V'$. Hence, $v_i^t$ has capacity one and is assigned a better student $d_i^t$, contradiction.
    this  means that $x=u_i$ for some $i\in [\enn]$.
    Then, $\mu(u_i)=w_i$ since otherwise $u_i$ already obtains her most preferred school.
    Consequently, $y=v_i^t$ holds for some $t\in [\emm]$
    with $v_i\in e_t$.
    By construction, $v_i\notin V'$ and $v_i^t$ is only assigned a student, namely $d_i^t$ and will not form with~$u_i$ a blocking pair since $\pr[v_i^t]=0$, a contradiction.
    This  concludes that $\pr$ is a good capacity increase vector.
  \end{proof}

  \begin{claim}\label{claim:sumSTA-bounddeg-good->VC}
   \vcbackward
 \end{claim}
  
  \begin{proof}[Proof of claim~\ref{claim:sumSTA-bounddeg-good->VC}]
      \renewcommand{\qedsymbol}{$\diamond$}
      Let $\pr$ be a good capacity increase vector, and let $\mu$ be a stable matching that assigns every student after increasing the capacities according to $\pr$.
      Let $V'=\{v_i \in V \mid \mu(e_t)=v_i^t \text{ for some } e_t \text{ with } v_i\in e_t\}$.
      We claim that $V'$ is a vertex cover of size at most $|\pr|_1-\emm$.
      Clearly, $V'$ is a vertex cover since $\mu$ matches every student, including the edge students.
      This means that there are at least $\emm$ seats are created to accommodate the edge students. 
      To see why $|V'|\le |\pr|_1-\emm$, we observe that for each vertex~$v_i\in V'$,
      there must exist a school~$v_i^t$ such that $\mu(e_t)=v_i^t$.
      However, since $v_i^t$ prioritizes $u_i$ over $e_t$,
      student~$u_i$ must be rematched to some school of the form $v_i^{z}$.
      If there would be more than $|\pr|_1-\emm$ vertices in $V'$, then we would need more than $|\pr|_1-\emm$ additional seats to accommodate the students~$u_i$ with $v_i\in V'$,
      a contradiction.
    \end{proof}
  The correctness follows from the above two claims.
 \end{proof}

In the next part, we show that \sumsta\ is highly inapproximable.

\newcommand{\sumstainapprox}{%
  \sumsta\ does not have any constant-factor approximation algorithm unless P=NP.
  This  holds even if the preference and priority lists are derived from a master list.
}
\newcommand{\sumstaleminapprox}{
  For each constant $d\ge 1$,
  if we can find a good capacity increase vector~$\pr$ with $\lone\le d(k+1)\en$ then we can find a set cover using at most $d(k+1)-1$ sets.
}
\begin{theorem}
\label{thm:sumsta-inapprox}
\sumstainapprox
\end{theorem}

\begin{proof}

We reduce from the NP-hard \scov\ problem.
\optprob{\scov}{
$\hat{m}$ subsets $\mathcal{C}=\{C_1,\dots, C_{\hat{m}}\}$ over $\en$ elements $\mathcal{X}=[\en]$.
}
{
A subfamily~$\mathcal{F}$ of $\mathcal{C}$ with minimum cardinality $|\mathcal{F}|$ such that $\cup_{C_i\in \mathcal{F}}C_i=[\en]$?
}
We remark that Dinur and Steurer \cite{dinur2014analytical} show that the optimum value is not approximable within any constant factor in polynomial-time, unless P=NP. In other words, given an instance $I$ of \scov, where the optimum value is $k$, it is still NP-hard to find a set cover using at most $dk$ sets for any constant $d$.

Let $I$ be an instance of \scov, with $OPT(I)=k$. We create an instance $I'$ of (the optimization variant of) \sumsta\ as follows.
 For each set $C_j$, we create a set school $c_j$ with capacity 1, along with a dummy student $d_j$.
 For each element $i\in [\en]$, we add an element student $e_i$.
For each $j\in [\m]$, we add $n$ students $u_j^1,\dots,u_j^{\en}$ along with schools $w_j^1,\dots,w_j^{\en}$ with capacity 1.

 Let $U_j\coloneqq \{ u_j^1,\dots,u_j^{\en}\}$.
  The preferences and priorities are defined as follows:

  \begin{tabular}{l@{\quad}|@{\quad}l}
      Students & Schools\\
       $e_i \colon \seqq{\mathcal{C}(e_i)}$ &  $c_j \colon d_j\succ \seqq{U_j}\succ \seqq{\mathcal{X}(c_j)}$ \\
    $u_j^{\ell} \colon c_j\succ w_j^{\ell}$ & $w_j^{\ell} \colon u_j^{\ell}$ \\
    $d_j \colon c_j$, &
    \end{tabular}

   \noindent where for each set $c_j$, $\mathcal{X}(c_j)$ denotes the set of students corresponding the elements in $C_j$, and for each student $e_i$, $\mathcal{C}(e_i)$ denotes the set of schools corresponding to the sets which contain~$i$. 
  
It is straightforward to check that all preferences are consistent with the master list $c_1,\dots,c_{\m}$, $w_1^1,\dots,w_1^{\en}$, $\dots,w_{\m}^1,\dots,w_{\m}^{\en}$ for the students and $d_1,\dots,d_{\m},u_1^1,\dots,u_1^{\en}$, $\dots,u_{\m}^1$, $\dots,u_{\m}^{\en}$, $e_1,\dots,e_{\en}$ for the schools.
Each school has initial capacity one.

\begin{claim}
\label{claim:com-HR-minsum-inapprox2}
\sumstaleminapprox
\end{claim}
\begin{proof}
\renewcommand{\qedsymbol}{$\diamond$}
Let $\pr$ be a good capacity vector such that $|\pr |_1\le d(k+1)\en$ and let $\mu$ be a stable matching that assigns every student. For each $c_j$, $d_j$ is matched to it, because they are mutually best for each other. Hence, the original capacity 1 is filled with $d_j$ for each $c_j$. Let $l$ denote the number of schools, where there is a student $e_i$ assigned. As for each $c_j$ all of $u_j^1,\dots,u_j^{\en}$ is better than any $e_i$ student, for each such school, all of their corresponding $u_j^{\ell}$ students must be assigned to them. As $\mu$ assings all $e_i$ students, we get that the capacity increase is at least $\en l+\en$. As $\en l+\en \le d(k+1)\en$, we get that $l\le d(k+1)-1$. As all element agents are matched in $\mu$, the $l$ sets that corresponds to these schools form a set cover.
\end{proof}

\begin{claim}
\label{claim:optinI'}
    For the optimum capacity increase vector $\pr$ in $I'$, we have that $\lone \le k\en +\en$.
\end{claim}
\begin{proof}
Let it be $C_{j_1},\dots,C_{j_k }$ be an optimal set cover. Define $\pr$ as follows: $\pr [w_j^{\ell}]=0$ for all $j\in [\m],\ell \in [\en]$. If $j\in \{ j_1,\dots,j_{k}\}$, $\pr [c_j]=\en+p$, where $p$ is the number of elements $i$ such that $C_j$ is the set with smallest index among $\{ j_1,\dots,j_{k} \}$ that covers $i$. Finally, $\pr [c_j]=0$ otherwise. Clearly, $|\pr |_1\le k\en +\en$.

   We claim that the matching $\mu$, given by the edges $\{ (d_j,c_j) \mid j\in [\m]\} \cup \{ (u_j^{\ell},w_j^{\ell}) \mid j\in [\m]\setminus \{ j_1,\dots ,j_{k}\} ,\ell \in [\en] \} \cup \{ (u_j^{\ell},c_j), (e_i,c_{l_i}) \mid j\in \{ j_1,\dots,j_{k}\} ,i,\ell \in [\en] \}$ is stable, where $l_i$ is the smallest index among $\{ j_1,\dots,j_{k}\}$ such that $i\in C_{l_i}$. $\mu$ is clearly feasible for $\quota +\pr$.

    Suppose there is a blocking pair $(x,y)$ to $\mu$. Student $x$ cannot be $d_j$, as $d_j$ is at her best choice. She also cannot be $u_j^{\ell}$, because for each $j\in [\m]$, either $u_j$ is at her best choice $c_j$ if $j\in \{ j_1,\dots,j_{k}\}$ or her only better school $c_j$ have capacity 1 in $\quota +\pr$, and is filled with a better student $d_j$. Finally, $x$ cannot be $e_i$, because again either $e_i$ is at her best school, or her better schools all have capacity 1 and filled with a better student $d_j$. This  is because $C_{j_1},\dots,C_{j_{k}}$ was a set cover, and for each $e_i$ we matched her the to best school ( which is the one with smallest index) who increased its capacity. Hence, there is no possible choice for a blocking student $x$, contradiction. 
  \end{proof}

Now suppose we have a polynomial-time $d$-approximation algorithm for \sumsta, for some constant $d$. Then by \cref{claim:optinI'}, we can find a good capacity increase vector $\pr$, with $|\pr |_1\le d(k+1)\en$.
From \cref{claim:com-HR-minsum-inapprox2}, we get that we can find a set cover using at most $d(k+1)-1$ sets. As $\frac{d(k+1)-1}{k}\le 2d$, this  algorithm gives us a polynomial $2d$-approximation for the set cover problem, which is a contradiction, if $P\ne NP$.
\end{proof}

Next, we show that \sumsta\ is parameterized intractable. 

\newcommand{\sumstawoneh}{%
  \sumsta\ is W[1]-hard wrt.\ the capacity bound~$\sumcap$.  
}
\begin{theorem}\label{sum-stable:w1h-cap}
  \sumstawoneh
\end{theorem}

  \begin{proof}[
 We provide a parameterized reduction from \MCC, which is W[1]-hard wrt.\ the solution size~$h$.

  \decprob{\MCC (\MCCs)}
  {An undirected graph~$G=(V, E)$, a number~$h\in \mathds{N}$, and a partition~$(V_1,\ldots,V_h)$ of the vertices~$V$.}
  {Is there a \emph{multi-colored $h$-clique}, i.e., a complete $h$-vertex subgraph containing exactly one vertex from each set~$V_i$?}

  Let $I=(G, h, (V_1,\ldots,V_h))$ be an instance of \MCCs, where $G=(V, E)$
  with $V=\{v_1,\ldots,v_{\enn}\}$ and $E=\{e_1,\ldots, e_{\emm}\}$.
  We create an instance~$I'$ of \sumsta\ as follows.
  For each pair of colors~$p=\{z, z'\}\subseteq [h]$, create an \myemph{edge-selector student}~$s_p$ and let $E^{p}$ denote the set consisting of the edges whose endpoints belong to $V_z$ and $V_{z'}$, respectively.
  For each edge~$e_{t}\in E$, create a school~$e_{t}$ and a student~$f_{t}$.
  For each vertex~$v_i\in V$, create a  school~$w_i$ and a student~$v_i$.
  Define $U=\{s_{p} \mid p\subseteq [h]\text{ with } |p|=2\}\cup \{f_t\mid e_t \in E\}\cup V$ and $W=E\cup \{w_i\mid v_i\in V\}$.
  In total, we created $\binom{h}{2}+\emm+\enn$ students and $\emm+\enn$ schools.

  The preference lists and the priority orders are defined as follows, where $p$ denotes an arbitrary color pair from $[h]$ and $e_{t}\in E^p$ denotes an edge with endpoints~$v_i$ and~$v_j$.

  {  \centering
    \begin{tabular}{@{}l|l@{}}
      Students & Schools\\
      $f_{t}\colon e_{t}$  & $e_{t}\colon f_{t} \succ \seqq{\{v_i,v_j\}} \succ  s_p$ \\          
       $v_{i}\colon \seqq{\{e_{z}\mid e_{z}\in E \text{ such that } v_i\in e_z\}} \succ w_i$ & $w_{i}\colon v_{i}$\\
       $v_{j}\colon \seqq{\{e_{z}\mid e_{z}\in E \text{ such that } v_j\in e_z\}} \succ w_j$ & $w_{j}\colon v_{j}$\\
        $s_p\colon \seqq{E^p}$ & \\
    \end{tabular}
    \par}
  
  For the sake of reasoning, if $X$ denotes a subset of edges, then the order~$\seqq{X}$ is according to the increasding order of the indices of the edges in $X$.   
  Initially, each school has capacity one. The total capacity increase bound is set to $\sumcap = \binom{h}{2}+h$.
  this  completes the construction of the instance, which can be done in polynomially time.
  
  It remains to show the correctness, i.e., $I$ admits a multi-colored $h$-clique if and only if there is a capacity vector~$\pr$ with $|\pr|_1\le \sumcap$ such that $(U,W,(\succ_x)_{x\in U\cup W}, \pq+\pr)$ admits a stable and perfect matching.
    For the ``only if'' part, let $K$ be a multi-colored $h$-clique of $I$.
  We claim that we can increase the capacities of the schools corresponding to the edges in $E(K)$ with total capacity increase $\sumcap$ to obtain a perfect and stable matching.
  We first define the desired matching~$\mu$, according to which we then define the capacity increase vector.
  For each ``clique-vertex''~$v_i\in V(K)$, let $e_t$ be an edge among all incident ``clique-edges'' $E(K)\cap \{e_z\in E\mid v_i\in e_z\}$ with the smallest index,
  and set $\mu(v_i)=e_t$.
  For each color pair~$p\subseteq [h]$, let $e_t\in E(K)$ be the ``clique-edge'' whose two endpoints have colors in~$p$, and set~$\mu(s_p)=e_t$.
  Finally, for each remaining vertex~$v_i\in V\setminus V(K)$ (resp.\ each edge~$e_t\in E$), set $\mu(v_i)=w_i$ (resp.\ $\mu(f_t)=e_t$).
  Clearly, the resulting matching~$\mu$ is perfect. The capacity increase vector is defined according to~$\mu$.
  More precisely, for each~$e_t\in E(K)$, define $\pr[e_t]=\mu^{-1}(e_t)-1$, and for each remaining school~$x\in W\setminus E(K)$, define $\pr[x]=0$.
  Now, observe that exactly $\binom{h}{2}$ schools increase their capacities from one to at least two.
  These $\binom{h}{2}$ schools admit in total $\binom{h}{2}+\binom{h}{2}+h$ students, including $\binom{h}{2}$ edge-selector students and $h$ students corresponding to the clique vertices.
  Hence, the total capacity increase is indeed $\binom{h}{2}+h$.
  It remains to show that $\mu$ is stable.
  Clearly, all students~$f_t$, $t\in [\emm]$, receive their best choice, so they do not form a blocking pair with any school.
  Neither can any student~$s_p$, $p\subseteq [h]$, be involved in any blocking pair since every school acceptable to~$s_p$ is full and ranks her at the last position.
  Finally, no student~$v_i\in V$ can be involved in any blocking pair since every school that $v_i$ prefers to~$\mu(v_i)$ has capacity one and is already assigned to its most preferred student.
  
  For the ``if'' part, let $\pr$ be a capacity increase vector with $|\pr|_1\le \sumcap$ such that $(U,W, \succ_{x\in U\cup W}, \pq+\pr)$ admits a stable and perfect matching, called~$\mu$.
  First, by stability, $\mu(f_t)=e_t$ holds for all $e_t\in E$ since $f_t$ and $e_t$ are each other's best choice.
  Second, since $\mu$ is perfect, each edge-selector student~$s_p$ has to be matched with some~$e_t\in E^p$.
  Since no two students from $\{s_p\mid p\subseteq [h]\text{ with } |p|=2\}$ have a common acceptable school,
  it follows that at least $\binom{h}{2}$ schools from $E$ have to increase their capacities.
  Moreover, for each $v_i\in \mu(s_p)$, in order to avoid school~$\mu(s_p)$ from forming a blocking pair with~$v_i$ it must hold that $\mu(v_i)\in \{e_z\in E\mid v_i\in e_z\}$.
  This  means that the students corresponding to the endpoints of $\mu(s_p)$ must be matched with some school from $E$.
  By the capacity bound, there can be at most $h$ such students, meaning that $V'=|\{v_i, v_j\mid \{v_i,v_j\}=\mu(s_p) \text{ for some } s_p\}|\le h$.
  This  can happen only if $V'$ forms a clique~$K$ of size $h$.
  That $K$ is multi-colored follows from the fact that each edge-selector students correspond to a distinct pair of colors.

  Since $\sumcap =\binom{h}{2}+h$, it immediately follows that \sumsta\ is W[1]-hard wrt.~$\sumcap$.
\end{proof}

\subsubsection{Algorithmic results}

On the positive side, we give an Integer programming formulation for \sumsta.

\begin{figure}[h]
\label{eq:IP}
\begin{alignat}{3}
\nonumber  \min \sum_{w\in W} r_w & \qquad& &\text{ subject to }\\
  |U|\cdot \sum_{w'\mid w'\succeq_u w}x_{(u,w')}+\sum_{u'\mid u'\succ_w u}x_{(u',w)} \ge \quota [w] + r_w & & &
    \forall (u,w)\in E \\
 \sum_{w\in \acset (u)}x_{(u,w)}=1 &&& \forall u\in U \\
 \sum_{u\in \acset (w)}x_{(u,w)}\le \quota [w] +r_w &&& \forall w\in W \\
     x_{(u,w)} \in \{0,1\}, &&& \forall (u,w)\in E \label{IP:unassign} \\
     r_w\in \mathbb{N} &&& \forall w\in W
\end{alignat}
\caption{An IP formulation for \sumsta.}\label{fig:IP}
\end{figure}

\begin{lemma}
\label{claim:IP}
The optimal solution to the Integer Program in \cref{fig:IP} gives an optimal solution for \sumsta. 
\end{lemma}

  \begin{proof}

Let $\pr$ be an optimal capacity increase vector and let $\stuopt$ be the student-optimal stable matching of $(U,W,(\succ_x)_{x\in U\cup W},\quota +\pr )$. Set $r_w = \pr [w]$ and $x_{(u,w)}=1$, if $(u,w)\in \stuopt$ and $0$ otherwise. Let $(u,w)\in E$. Then, by the stability of $\stuopt$, either $u$ is at a school that is at least as good as $w$ or $w$ is filled with students strictly better than $u$. As $\quota [w]+\pr [w]\le n$ for any $w\in W$ (there can be at most $n$ students at any school), the first inequality is satisfied. As every student is matched and $\stuopt$ is feasible, the second and third inequalities are also satisfied, hence $(\textbf{x},\textbf{r})$ is a solution to the IP and $\sum_{w\in W}r_w = \lone$. So the optimum of the IP is at most $\lone$.

In the other direction, take an optimal solution $(\textbf{x},\textbf{r})$ of the Integer Program. Set $\pr [w] =r_w$ and $\mu [u] = w$, if $x_{(u,w)}=1$. By the second and third inequalities, this gives a feasible matching $\mu$. Suppose that an edge $(u,w)$ blocks $\mu$. Then, $\sum_{w'\mid w'\succeq_u w}x_{(u,w)}=0$ and $\sum_{u'\mid u'\succ_w u}x_{(u,w)}< \quota [w] +r_w$, so the first inequality is not satisfied, contradiction. As $\sum_{w\in W}r_w = \lone$, we get that the optimum of the IP gives an optimal solution to \sumsta.
  


\end{proof}

\noindent Based on \cref{lemma:minsum-just-envy}, we can also give a simple greedy approximation algorithm.

\begin{algorithmic}
 \label{alg:Unasapprox}
\begin{algorithm}

    \caption{$\lceil |\Unas |/c \rceil$-approximation}
\STATE Input: An instance $I$ of \sumsta\
\STATE $\mu :=\emptyset$
\STATE $\Unas:=$ unmatched students in the student-optimal stable matching $\stuopt$
\STATE Set $L:= \Unas$
\STATE Delete the students $u\in \Unas$ from $I$
\WHILE{$L\ne \emptyset$}
\STATE Choose the next (at most) $c$ elements in $L$, add them to $I$ and define $\Vecs$ as in Lemma~\ref{lemma:minsum-just-envy}
\STATE Set $\Unas$ to be the unmatched students of the student-optimal stable matching in $I$
\FOR{all $\pv\in \Vecs$}
\STATE Compute $I^*(\pv),\mnew $ with Algorithm~\ref{alg:procedure} if it exists
\ENDFOR
\STATE Let $\pv \in \Vecstable$ be the vector where the smallest aggregate capacity increase is needed
\STATE Update the capacities according to $I^*(\pv)$ and let $I:=I^*(\pv)$
\STATE $\mu :=\mnew (\pv)$
\ENDWHILE
\STATE \textbf{Return} $\mu$
\end{algorithm}
\end{algorithmic}

\begin{theorem}\label{thm:sumsta-unas-approx}
   \sumsta\ admits an $|\Unas|$-approx.\ algorithm. Furthermore, it admits a polynomial-time $\lceil |\Unas |/c \rceil$-approximation algorithm for any constant $c$. 
\end{theorem}

\begin{proof}



We show that Algorithm 2 
produces an $\lceil |\Unas |/c \rceil$-approximation in polynomial time for any fixed constant $c$. Clearly, in each iteration of the while loop, we have that there are at most $c$ unassigned students in the actual student-optimal stable matching and hence the for the set $\Vecs$, we have that $|\Vecs |\le \studlen^c$. Hence, the running time is $\studlen^c\cdot poly (|E|)$.

To see that this gives a $\lceil |\Unas |/c \rceil$-approximation, observe that in any iteration of the while loop, the additional number of seats required is at most $OPT$. Indeed, at the start of each iteration in the while loop, the capacities are weakly larger than in $I$, hence we know that there is a capacity increase vector $\pr$ with $\lone \le OPT$ such that with capacities $\quota + \pr$ every student can be matched in a stable matching. 

Therefore, the total capacity increase of the algorithm is at most $OPT\cdot \lceil |\Unas |/c \rceil$, since there are $\lceil |\Unas |/c \rceil$ iterations in the while loop.


\end{proof}

By checking all possible assignment vector, we obtain the following simple result.
\begin{theorem}\label{thm:sumsta-XP-fpt-unassiged-len}
\sumsta\ can be solved in $\studlen^{|\Unas|}\cdot (n+m)^{O(1)}$ time and hence is FPT wrt.\ $(|\Unas|,\studlen)$. 
\end{theorem}
\begin{proof}
  From Lemma \ref{lemma:minsum-just-envy}, we can see that to compute the optimum value (and an optimal capacity increase vector), it is enough to iterate over $\Vecs$, compute $\mnew (\pv)$ if it exists and then choose the best among the matchings that are stable.
  For each $\pv\in \Vecs$, $\je(\pv)$ can be computed in $O(|E|)$ time, where $E$ is the set of all acceptable pairs.
  Furthermore, $|\Vecs|\le \studlen^{|\Unas|}$. 
  Hence, we can iterate through all $\pv\in \Vecs$ and find the optimal solution in time $O(\studlen^{|\Unas|}\cdot |E|)$. 
\end{proof}

\subsection{Minimax capacities}
Although the \minsum\ version is NP-hard and inapproximable, we show that the \minmax\ version can be solved in polynomial-time.

\begin{algorithm}[h]
\begin{algorithmic}

\caption{Algorithm for \maxsta.}
\label{alg:minmax}

    \STATE $\pr [w]:=0$ for all $w\in W$
    \STATE $\mu:=$ student-optimal stable matching
    \WHILE{$\mu $ does not match all students}
    \STATE $\pr [w]=\pr [w]+1$ for all $w\in W$
     \STATE $\mu:=$ student-optimal stable matching with capacities $\quota +\pr$
    \ENDWHILE
\end{algorithmic}
\end{algorithm}
\begin{theorem}
\label{thm:maxsta}
\maxsta\ can be solved, and the corresponding student-optimal stable and perfect matching can be found, in polynomial time.
\end{theorem}
\begin{proof}
We show that Algorithm \ref{alg:minmax} returns a student-optimal solution, such that $\lmax$ is minimum. 
For this, we use that increasing a school's capacity weakly improves every student's situation by lemma \ref{lem:capincrease}. 
Hence, if $\pr$ is the optimal such capacity increase vector, with $\lmax=\maxcap$, then it must hold, that if we have $\pr [w]=\maxcap$, for all $w\in W$, then the student-optimal stable matching still matches all students and is best for the students for any capacity increase vector with $\lmax\le \maxcap$. Hence, the algorithm finds the optimal value $\maxcap$, and a student-optimal stable matching for it.
\end{proof}

\section{Stable and Efficient Matchings}\label{sec:SE}

In this  section, we consider finding optimal capacity vectors for stable and efficient matchings and show this  problem is NP-hard, hard to approximate, and parameterized intractable.

\subsection{Minisum capacities}

First, we give an example which will be crucial for the hardness reduction. 
\begin{example}
\label{ex:stable-eff}
The following example admits no efficient and stable matching, even if we allow to increase one capacity by one.
We have five students $u_1,\dots ,u_5$ and five schools $w_1,\ldots,w_5$ each with capacity one, with the following preferences.

{\centering\begin{tabular}{ll|ll}
    $u_1:$ & $w_1\succ w_3\succ w_4$ &  $w_1:$  & $u_5\succ u_3\succ u_2\succ u_1$ \\
    $u_2:$ & $w_1\succ w_2$ & $w_2:$ & $u_2\succ u_5\succ u_3\succ u_4$ \\
    $u_3:$ & $w_2\succ w_1\succ w_3$ & $w_3:$ & $u_3\succ u_4\succ u_1\succ u_5$ \\
    $u_4:$ & $w_2\succ w_3\succ w_5$ & $w_4:$ & $u_1$ \\
    $u_5:$ & $w_3\succ w_2\succ w_1$ & $w_5:$ & $u_4$  \\
    
\end{tabular}
\par}

  If we increase none of the capacities or only increase school $w_4$'s or school $w_5$'s capacity by 1, then the student-optimal stable matching is $\mu =\{ (u_5,w_1),(u_2,w_2),(u_3,w_3),(u_1,w_4),(u_4,w_5)\}$, which is Pareto-dominated by the matching $\mu'=\{ (u_2,w_1),(u_5,w_2),(u_3,w_3),(u_1,w_4),$ $(u_4,w_5)\}$.

If we increase the capacity of school $w_1$ by 1, then the student-optimal stable matching becomes $\mu =\{ (u_2,w_1),(u_3,w_1),(u_5,w_2),$ $(u_4,w_3),(u_1,w_4)\}$, which is Pareto-dominated by $\mu'=\{ (u_2,w_1),$ $(u_3,w_1),(u_4,w_2),$ $(u_5,w_3),(u_1,w_4)\}$.

If we increase the capacity of school $w_2$ by 1, then the student-optimal stable matching is $\mu =\{ (u_2,w_1),(u_3,w_2),$ $(u_5,w_2),(u_4,w_3),$ $(u_1,w_4)\}$, which is Pareto-dominated by $\mu'=\{ (u_2,w_1), (u_3,w_2),$ $(u_4,w_2),(u_5,w_3),(u_1,w_4)\}$.

Finally, if we increase the capacity of school $w_3$ by 1, then the student-optimal stable matching is $\mu=\{ (u_5,w_1),(u_2,w_2),(u_3,w_3), $ $(u_4,w_3),$ $(u_1,w_4)\}$, which is Pareto-dominated by  $\mu'=\{ (u_2,w_1),$ $(u_3,w_2),(u_4,w_3),$ $(u_5,w_3),(u_1,w_4)\}$.

However, if student~$u_5$ is deleted, then $\mu =\{ (u_2,w_1)$, $(u_3,w_2)$, $(u_4,w_3),(u_1,w_4)\}$ is stable, perfect, and efficient, without any capacity increase needed.
\end{example}

Using the gadget given in \cref{ex:stable-eff} and a construction similar to the one for \cref{sum-stable:w1h-cap}, we can show hardness for \sumse.
\newcommand{\thmsumsewoneh}{
  \sumse\ is NP-complete and W[1]-hard wrt.\ the capacity bound~$\sumcap$.
}
\begin{theorem}\label{sum-se:w1h-cap}
  \thmsumsewoneh
\end{theorem}

\begin{proof}
  To show hardness, we provide a polynomial-time parameterized reduction from the NP-complete problem \MCC, which is W[1]-hard with respect to the solution size~$h$. Let $I=(G, h, (V_1,\ldots,V_h))$ be an instance of \MCCs, where $G=(V, E)$
  with $V=\{v_1,\ldots,v_{\enn}\}$ and $E=\{e_1,\ldots, e_{\emm}\}$.
  We create an instance~$I'$ of \sumsta\ as follows.
  For each pair of colors~$p=\{z, z'\}\subseteq [h]$, create an \myemph{edge-selector gadget}~$G_p$. These gadgets are copies of \cref{ex:stable-eff} with 5 schools $w_1^p,\dots,w_5^p$ and 5 students $u_1^p\dots, u_5^p$. Let us call the fifth student $u_5^p$ as $s_p$, who will be the \myemph{edge-selector student} for color pair $p$. The preferences are the same, except for each such gadget $G_p$, $s_p$ has other acceptable schools too.
  For each edge~$e_{t}\in E$, create a school~$e_{t}$ and a student~$f_{t}$.
  For each vertex~$v_i\in V$, create a  school~$w_i$ and a student~$v_i$.
  Define $U=\{u_1^p,u_2^p,\dots, u_5^p=s_p \mid p\subseteq [h]\text{ with } |p|=2\}\cup \{f_t\mid e_t \in E\}\cup V$ and $W=E\cup \{w_i\mid v_i\in V\} \cup \{ w_1^p,w_2^p,\dots,w_5^p \mid p\subseteq [h]\text{ with } |p|=2\}$.

  Let $p=\{z,z'\}\subseteq [h]$ be a pair of colors and $E^{p}$ denote the set consisting of the edges whose endpoints belong to $V_z$ and $V_{z'}$, respectively.
  The preference lists of the students and the priority orders of the schools are defined as follows,
  where 
  $e_{t}\in E^p$ denotes an edge with endpoints~$v_i$ and~$v_j$.
  The preferences and priorities of the local students and schools inside the $G_p$ gadgets are not included
  (except for $u_5^p=s_p$, whose preference list is extended non-locally), because they are inherited from \cref{ex:stable-eff}.
  
  {
   \centering
    \begin{tabular}{l@{\quad}|@{\quad}l}
      Students & Schools\\
      $f_{t}\colon e_{t}$  & $e_{t}\colon f_{t} \succ \seqq{\{v_i,v_j\}} \succ  s_p$ \\          
       $v_{i}\colon \seqq{\{e_{z}\mid e_{t}\in E \text{ s.t. } v_i\in e_z\}} \succ w_i$ & $w_{i}\colon v_{i}$\\
       $v_{j}\colon \seqq{\{e_{z}\mid e_{t}\in E \text{ s.t. } v_j\in e_z\}} \succ w_j$ & $w_{j}\colon v_{j}$\\
        $s_p\colon \seqq{E^p} \succ w_3^p\succ w_2^p\succ w_1^p$ & \\
    \end{tabular}\par}

  Initially, each school has capacity one. The total capacity increase bound is set to $\sumcap = \binom{h}{2}+h$.
  this  completes the construction of the instance, which can be done in polynomially time.

    It remains to show the correctness, i.e., $I$ admits a multi-colored $h$-clique if and only if there is a capacity vector~$\pr$ with $|\pr|_1\le \sumcap$ such that $(U,W,(\succ_x)_{x\in U\cup W}, \pq+\pr)$ admits a stable, perfect and efficient matching.

  For the ``only if'' part, let $K$ be a multi-colored $h$-clique of $I$.
  We claim that we can increase the capacities of the schools corresponding to the edges in $E(K)$ with total capacity increase $\sumcap$ to obtain a stable, perfect and efficient matching.
  We first define the desired matching~$\mu$, according to which we then define the capacity increase vector.
  For each ``clique-vertex''~$v_i\in V(K)$, let $e_t$ be an edge among all incident ``clique-edges'' $E(K)\cap \{e_z\in E\mid v_i\in e_z\}$ with the smallest index,
  and set $\mu(v_i)=e_t$.
  For each color pair~$p\subseteq [h]$, let $e_t\in E(K)$ be the ``clique-edge'' whose two endpoints have colors in~$p$, and set~$\mu(s_p)=e_t$. For the other students inside the gadgets we let $\mu (u_2^p)=w_1^p, \mu (u_3^p)=w_2^p, \mu (u_4^p)=w_3^p, \mu (u_1^p)=w_4^p$.
  Finally, for each remaining vertex~$v_i\in V\setminus V(K)$ (resp.\ each edge~$e_t\in E$), set $\mu(v_i)=w_i$ (resp.\ $\mu(f_t)=e_t$).
  Clearly, the resulting matching~$\mu$ is perfect. The capacity increase vector is defined according to~$\mu$.
  More precisely, for each~$e_t\in E(K)$, define $\pr[e_t]=|\mu^{-1}(e_t)|-1$, and for each remaining school~$x\in W\setminus E(K)$, define $\pr[x]=0$.
  Now, observe that exactly $\binom{h}{2}$ schools increase their capacities from one to at least two.
  These $\binom{h}{2}$ schools admit in total $\binom{h}{2}+\binom{h}{2}+h$ students, including $\binom{h}{2}$ edge-selector students and $h$ students corresponding to the clique vertices.
  Hence, the total capacity increase is indeed $\binom{h}{2}+h$.
  It remains to show that $\mu$ is stable and efficient.

  As we have seen in \cref{ex:stable-eff}, $\mu$ restricted to $\{ u_1^p,\dots,u_4^p,w_1^p,\dots,w_5^p\}$ is stable and efficient for any color pair $p$. Also, none of them has any acceptable schools/students outside their corresponding gadget, so there can be no matching, where one of them strictly improves and the others weakly improve.
  All students~$f_t$, $t\in [\emm]$, receive their best choice, so they do not form a blocking pair with any school. It follows that they cannot strictly improve either.
  The students~$s_p$, $p\subseteq [h]$, cannot be involved in any blocking pair too, since every school $e_t$ that is better for~$s_p$ has capacity one and is filled with a better student $f_t$. As $f_t$ has only one acceptable school $e_t$, it follows that $s_p$ could only strictly improve, if $f_t$ is dropped from $e_t$, which would imply that she cannot weakly improve. 
  Finally, no student~$v_i\in V$ can be involved in any blocking pair since every school that $v_i$ prefers to~$\mu(v_i)$ has capacity one and is already assigned to its most preferred student. Similarly to the previous case, this  only student has no other acceptable schools, hence a $v_i$ student could only strictly improve, if an $f_t$ student is dropped out from her school. 

  Therefore, we obtained that $\mu$ is stable, and a student could only strictly improve, if an another student obtains a worse position, hence it is also Pareto-efficient.
  
  For the ``if'' part, let $\pr$ be a capacity increase vector with $|\pr|_1\le \sumcap$ such that $(U,W, \succ_{x\in U\cup W}, \pq+\pr)$ admits a stable, perfect and efficient matching, called~$\mu$.
  First, by stability, $\mu(f_t)=e_t$ holds for all $e_t\in E$ since $f_t$ and $e_t$ are each other's best choice.
  Second, since $\mu$ is perfect, stable and efficient, each edge-selector student~$s_p$ has to be matched with some~$e_t\in E^p$ or if not, then the sum of the capacity increases in the schools of gadget $G_p$ is at least 2. 
  Let the number of edge selector students who are assigned to a school $e_t\in E^p$ be $\ell$.
  Since no two students from $\{s_p\mid p\subseteq [h]\text{ with } |p|=2\}$ have a common acceptable school,
  it follows that at least $\ell$ schools from $E$ have to increase their capacities. 
  Moreover, for each $v_i\in \mu(s_p)$, in order to avoid school~$\mu(s_p)$ from forming a blocking pair with~$v_i$ it must hold that $\mu(v_i)\in \{e_z\in E\mid v_i\in e_z\}$.
  this  means that the students corresponding to the endpoints of $\mu(s_p)$ must be matched with some school from $E$.

  Summarizing, we get that the sum of increase in the capacities is at least $2(\binom{h}{2}-\ell )+\ell +\frac{1+\sqrt{1+8\ell}}{2}$, because there are $\binom{h}{2}-\ell$ gadgets, where the sum of increases is at least two, and the other $\ell$ edge selector students are matched to $\ell$ schools from $E$, and $\ell$ edges are incident to at least $\frac{1+\sqrt{1+8\ell}}{2}$ vertices (as $x$ vertices can only span $\binom{x}{2}$ edges).  As the derivative of $-\ell +\frac{1+\sqrt{1+8\ell}}{2}$ is $-1+\frac{2}{\sqrt{1+8\ell}}<0$, whenever $\ell \ge 1$, we get that $-\ell +\frac{1+\sqrt{1+8\ell}}{2}$ is strictly decreasing over the interval $[1,\binom{h}{2}]$. Therefore, it is minimal at $\ell =\binom{h}{2}$ or $\ell =0$. Also, if $\ell =\binom{h}{2}$, then $2(\binom{h}{2}-\ell )+\ell +\frac{1+\sqrt{1+8\ell}}{2}=\binom{h}{2}+h=\sumcap$ and if $\ell=0$, then the needed capacity increase is at least $2\binom{h}{2}>\sumcap$, as $h\ge 2$. Hence, all edge selector students must be matched to $e_t$ schools, and the corresponding $\binom{h}{2}$ edges must be adjacent to exactly $h$ vertices. This  can happen only if these vertices form a clique~$K$ of size $h$.
  That $K$ is multi-colored follows from the fact that each edge-selector students correspond to a distinct pair of colors.

  This concludes the proof for NP-hardness result. Since $\sumcap=\binom{h}{2}+h$, it immediately follows that \sumse\ is W[1]-hard wrt.\ $\sumcap$.
\end{proof}

Note that in the proof of Theorem~\ref{sum-se:w1h-cap}, the constructed instance has a perfect and stable matching.
It follows that \sumse\ is NP-hard even if $|\Unas |=\studlen =0$. 

\begin{corollary}
\sumse\ cannot be $\studlen$-approximated and cannot be $|\Unas|$-approximated if $P\ne NP$.
\end{corollary}

\begin{remark}\label{rem:SumSE-hard-preflength}
By using \cref{ex:stable-eff} as a gadget with the proof of \cref{thm:sumsta-inapprox} instead in an analogous way (substitute the originally unassigned students $e_i$-s with the gadget of \cref{ex:stable-eff}), we can show that \sumse\ remains NP-complete with constant preference and priority lengths too. (We also use the fact that \scov\ remains NP-hard even if each set has size at most three and each element appears in exactly three sets, so the degrees will be bounded by a constant).
\end{remark}

\subsection{Minimax capacities}

We start with an example to illustrate the difficulties with \maxse\ opposed to \maxsta.

\begin{example}
While \maxsta\ can be solved easily with \cref{alg:minmax}, it does not give an optimal solution for \maxse. Consider the instance in \cref{ex:stable-eff}. Add 10 more schools $w_1',\dots w_5',v_1,\dots,v_5$. The $w_1',\dots,w_5'$ schools rank the $u_i$ students in the same way as $w_1,\dots,w_5$. Furthermore, for each $u_i$ we extend her preference list by first appending $v_i$ and then the $w_j'$ schools for each $j$ such that $w_j$ is acceptable to $u_i$, to the end of her preference list in the same order as the acceptable $w_j$ schools. Finally, for each of $\{ w_1,\dots,w_5,v_1,\dots,v_5\}$, add a dummy student, such that each school and its dummy student mutually consider each other best.
The preferences are (excluding the dummy students in the preference lists):
\begin{center}
\begin{tabular}{ll|ll}
    $u_1:$ & $w_1\succ w_3\succ w_4\succ v_1 \succ w_1'\succ w_3'\succ w_4'$ &  $w_1,w_1':$  & $u_5\succ u_3\succ u_2\succ u_1$ \\
    $u_2:$ & $w_1\succ w_2\succ v_2\succ w_1'\succ w_2'$ & $w_2,w_2':$ & $u_2\succ u_5\succ u_3\succ u_4$ \\
    $u_3:$ & $w_2\succ w_1\succ w_3\succ v_3\succ w_2'\succ w_1'\succ w_3'$ & $w_3,w_3':$ & $u_3\succ u_4\succ u_1\succ u_5$ \\
    $u_4:$ & $w_2\succ w_3\succ w_5\succ v_4\succ w_2'\succ w_3'\succ w_5'$ & $w_4,w_4':$ & $u_1$ \\
    $u_5:$ & $w_3\succ w_2\succ w_1\succ v_5\succ w_3'\succ w_2'\succ w_1'$ & $w_5,w_5':$ & $u_4$ 
    \\
    & & $v_i:$ $u_i$
    
\end{tabular}
\end{center}

Each capacity is one.
It is clear that there is no stable and efficient matching without capacity increases: each school with a dummy student must admit only its dummy student and $u_1,\dots,u_5$ cannot be assigned in a stable and efficient way to $w_1',\dots,w_5'$ as we have seen in \cref{ex:stable-eff}. There is an optimal capacity increase vector, where the maximum increase is one: increase each of $v_1,\dots,v_5$'s capacity by one and the others by zero. Then, if we assign $u_i$ to $v_i$ it will be both stable and Pareto-efficient. However, if we increase all school capacities by one, then all of $u_1,\dots,u_5$ must be assigned to a school among $w_1,\dots,w_5$. But as we have seen in \cref{ex:stable-eff}, it is not possible to do in a stable and efficient way.
\end{example}

Finally, we show that \maxse\ is also NP-complete and inapproximable, which is in contrast with our results for \maxsta.

\newcommand{\clmmaxseforward}{%
  If there is good capacity increase vector $\pr$ with $\lmax\le \eta$ or $|\pr|_1\le \eta$, then there is a satisfying assignment $f$. 
}
\begin{theorem}\label{thm:maxse}
    \maxse\ is NP-complete, even if $\maxcap =3$ and each preference and priority list has length bounded by a constant. Furthermore, neither \maxse\ nor \sumse\ admit a polynomial-time $\mathcal{O}((|U|+|E|)^{1-\varepsilon})$-approximation algorithm for any $\varepsilon >0$ if $P\ne NP$, even if the preference lists of the students are at most 5 long.
\end{theorem}
\begin{proof}
Before we start the proof, let us describe the two gadgets that we are going to use. Let $\eta \ge 3$ be a positive integer number specified later.

\smallskip
\noindent\textbf{Clause gadget.} We have 3 students $u_1,u_2,c$ and 3 schools $w_1,w_2,w_3$ along with $2\eta$ dummy students $d_1^1,\dots,d_1^{\eta},d_2^1,\dots,d_2^{\eta}$ and $2\eta$ dummy schools $s_1^1,\dots,s_1^{\eta},s_2^1,\dots,s_2^{\eta}$. The preferences are:

  {
    \centering\begin{tabular}{ll|ll}
    $u_1:$ & $w_1\succ w_2$ &  $w_1:$  & $c\succ u_1\succ d_1^1\succ \cdot \succ d_1^{\eta}\succ  u_2$ \\
    $u_2:$ & $w_1\succ w_2\succ w_3$ & $w_2:$ & $u_1\succ u_2\succ d_2^1\succ \cdots \succ d_2^{\eta} \succ c$ \\
    $c:$ & $w_2\succ w_1$ & $w_3:$ & $u_2$ \\
    $d_i^j:$ & $w_i\succ s_i^j$ & $s_i^j:$ & $d_i^j$, \\
\end{tabular}
\par}

for $i\in [2], j\in [\eta ]$.
The initial capacities of all schools are one. The crucial property of this  gadget is that if each school has capacity at most $\eta +1$ (and at least one), then in the student-optimal stable matching $c$ is matched to $w_1$ and either $u_1$ is matched to $w_2$ (if $w_1$ has capacity one) or $u_2$ is matched to $w_2$ (if $w_1$ has capacity more than one). In either case, the resulting matching is not Pareto-efficient, as student $c$ would mutually like to switch places with $u_1$ or $u_2$ (the one at $w_2$).

However, if student $c$ would be matched out from the gadget to a better school, than there is a stable and Pareto-efficient matching without any capacity increase needed, $\mu =\{ (u_1,w_1),(u_2,w_2)$, $(d_i^j,s_i^j)\mid i\in [2],j\in [\eta  ]\}$.

\medskip
\noindent \textbf{Variable gadget.} We have 5 students $p_1,p_2,p_3,T,F$ and 5 schools $z_1,z_2,z_3,z_4,z_5$ along with $2\eta $ dummy students $e_1^1,\dots,e_1^{\eta},e_2^1,\dots,e_2^{\eta}$ and $2\eta$ dummy schools $t_1^1,\dots,t_1^{\eta},t_2^1,\dots,t_2^{\eta}$. The preferences and priorities are:
\begin{center}
\begin{tabular}{ll|ll}
    $p_1:$ & $z_1\succ z_2$ &  $z_1:$  & $T\succ F\succ p_3\succ e_1^1\succ \cdots \succ e_1^{\eta}\succ  p_1\succ p_2$ \\
    $p_2:$ & $z_1\succ z_2\succ z_3$ & $z_2:$ & $p_1\succ e_2^1\succ \cdots \succ e_2^{\eta} \succ p_2\succ p_3$ \\
    $p_3:$ & $z_2\succ z_1\succ z_4$ & $z_3:$ & $p_2$ \\
    $T:$ & $z_1$ &   $z_4:$ & $p_3$\\
    $F:$ & $z_1\succ z_5$ &  $z_5:$ & $F$ \\
    $e_i^j:$ & $z_i\succ t_i^j$ & $t_i^j:$ & $e_i^j$, 
\end{tabular}
\end{center}
\noindent where $i\in [2], j\in [\eta ]$.

\noindent Each initial capacity is one. The crucial observation of this  gadget is that if both $T$ and $F$ would be matched outside the gadget to a better school, and each school inside has capacity at most $\eta +1$, then in the student-optimal stable matching, $p_3$ is matched to $z_1$ and $p_1$ is matched to $z_2$. Hence, $p_1$ and $p_3$ both envy the other one's school, so the matching is not Pareto-efficient. However, if one of $T$ and $F$ is matched inside the gadget, then the student-optimal stable matching without any capacity increase matches her to $z_1$ (in the case when both $T$ and $F$ is matched inside then $T$ is at $z_1$ and $F$ is at $z_5$), while it matches $p_1$ to $z_2$, $p_2$ to $z_3$, $p_3$ to $z_4$ and each $e_i^j$ to $t_i^j$ and it is efficient.

\smallskip
\noindent \textbf{The reduction.}
Now we provide our reduction from \rrsat.
\decprob{\rrsat}{
  A Boolean 3CNF formula~$\varphi = \{C_1\land \cdots \land C_{\emm}\}$ over a set~$X_1,\ldots,X_{\enn}$ of variables such that
  each clause~$C_j$ has three literals,
  and each variable $X_i$ appears exactly twice in positive and exactly twice in negated form.
}{
Does $\varphi$ admit a satisfying assignment?
}

Let $I=\varphi$ be an instance of \rrsat, with clauses $C_1,\dots,C_{\hat{m}}$ and variables $X_1,\dots,X_{\hat{n}}$. Note that $3\hat{m}=4\hat{n}$. We create an instance $I'$ of \maxse\ and an instance $I''$ of \sumse\ as follows.

\begin{compactitem}[--]
    \item For each clause $C_j$, $j\in [\hat{m}]$, we create a clause gadget $G_1^j$ . Let us denote the student $c$ in gadget $G_1^j$ by $c^j$.
    \item For each variable $X_i$, we create a variable gadget $G_2^i$. Let us call the $T$ and $F$ students in $G_2^i$ by $T^i$ and $F^i$.
    \item Finally, for each variable $X_i$, create two literal schools $x_i$ and $\overline{x_i}$ each with one dummy student $y_i$ and $\overline{y_i}$ respectively.
\end{compactitem}
In total, we have $(3+2\eta )\hat{m}+(7+2\eta )\hat{n}$ schools and $(3+2\eta )\hat{m}+(7+2\eta )\hat{n}$ students. As $3\hat{m}=4\hat{n}$, both these quantities are $\mathcal{O}(\eta \hat{n}).$

For each literal school $x_i$, its priority list is $y_i\succ T^i\succ c^1(x_i)\succ c^2(x_i)$, where $c^l(x_i)=c^j$, if and only if the $l$-th occurrence of literal $X_i$ is in $C_j$. For each literal school $\overline{x_i}$, its priority list is $\overline{y_i}\succ F^i\succ c^1(\overline{x_i})\succ c^2(\overline{x_i})$, where $c^l(\overline{x_i})=c^j$, if and only if the $l$-th occurrence of literal $\overline{X_i}$ is in $C_j$. The dummy students $y_i,\overline{y_i}$ only consider their corresponding school acceptable.

Every other student and school inherits its preference/priority list from her/its corresponding gadget, with the following extensions.  
For each gadget $G_1^j$, $j\in [\hat{m}]$, we extend $c^j$'s preference list by adding $\ell_{j_1}\succ \ell_{j_2}\succ \ell_{j_3}$ to the beginning of her preference list, where $\ell_{j_1},\ell_{j_2},\ell_{j_3}$ corresponds to the three literal schools corresponding to the literals in $C_j$. Finally, for each student $T^i$, we add $x_i$ as her best school and for each student $F^i$, we add $\overline{x_i}$ as her best school. 


Each initial capacity is one. Furthermore $\maxcap =3$ for $I'$ and $\sumcap =3\hat{n}$ for $I''$.

\begin{claim}
\label{claim:maxse1}
If there exists a satisfying assignment, then there is a good capacity increase vector $\pr$ with $\lone= 3\hat{n}$ and $\lmax=3$.
\end{claim}
\begin{proof}
\renewcommand{\qedsymbol}{$\diamond$}
Let $f$ be a satisfying assignment. First we describe $\pr$. For each $i\in [\hat{n}]$ if $X_i$ is set to True, then let $\pr [x_i]=3$ and $\pr [\overline{x_i}]=0$, otherwise $\pr [\overline{x_i}]=3$ and $\pr [x_i]=0$. For all other schools $s$, $\pr [s]=0$. Clearly, $|\pr|_1= 3\hat{n}$ and $\lmax=3$. Let $\mu$ be the following matching. Match each dummy student $y_i,\overline{y_i}$, $i\in [\hat{n}]$ to her dummy school.
If $X_i$ is set to True, then we match $T^i$ to school $x_i$, otherwise we match $F^i$ to school $\overline{x_i}$. Then, for each clause $C_j$, let $\ell_{i_j}$ be the first literal in $C_j$ that is true and match student $c^j$ to the literal school of $\ell_{i_j}$. $\mu$ is feasible for $\quota +\pr$, because each variable appears exactly twice in both positive and negative form, hence each literal school has only 4 adjacent student, including its dummy one. So, even if we match both its clause $c^j$ students to it, its capacity $1+3$ in $\quota +\pr$ is not violated.

As $f$ was a satisfying assignment, for each $j\in [\hat{m}]$, $c^j$ is matched outside of the gadget $G_1^j$ to a better school, hence as we have seen before, the rest of the students inside can be matched in a stable and Pareto-efficient way. As $f$ is consistent, for each $i\in [\hat{n}]$, exactly one of $T^i$ and $F^i$ is matched outside of $G_2^i$ to a better school, therefore as we have seen, the remaining students inside can be matched in a stable and Pareto-efficient way. 

We claim that the resulting matching is stable and Pareto-efficient. For stability, we show that no edge that is not induced by a gadget blocks. Clearly, no $(y_i,x_i),(\overline{y_i},\overline{x_i})$ edge blocks. No edge of type $(c^j,\ell_i)$, $\ell \in \{ x,\overline{x} \}$ can block, because we matched $c^j$ to the best school that increased its capacity and the other ones are filled with their best (dummy) student. Finally, no edge of type $(T^i,x_i)$ or $(F^i,\overline{x_i})$ can block, because either it is included in the matching, or if not, then the school is filled with its best student. For Pareto-efficiency, first note that the matching is Pareto-efficient inside the gadgets.
It is easy to check that for any student who has acceptable schools outside her gadget and is not at her best school, all her better schools are filled with students who are at their best school. Hence, a student could only improve, if another one disimproves, so the matching is also efficient. 
\end{proof}

\begin{claim}
\label{claim:maxse2}
\clmmaxseforward
\end{claim}

  \begin{proof}
\renewcommand{\qedsymbol}{$\diamond$}
Suppose there is a good capacity increase vector $\pr$ with $\lmax\le \eta$ and a stable and efficient matching $\mu$ with respect to capacities $\quota +\pr$. Then, it must hold that for each gadget $G_1^j$, $c^j$ is matched to a literal school, as otherwise there is no stable and Pareto-efficient matchings inside $G_1^j$ with capacity increases at most $\eta$. Furthermore, if $c^j$ is matched to a literal school $\ell_i$, $\ell \in \{ x,\overline{x} \}$, then because $\ell_i$ likes its corresponding $T^i$ or $F^i$ student better, who in return considers $\ell_i$ her best school, it must hold that she is also matched to $\ell_i$ in $\mu$, by the stability of $\mu$. Therefore, if there would be a variable $X_i$, such that both $x_i$ and $\overline{x_i}$ has a $c^j$ resident from a gadget $G_1^j$, then both $T^i$ and $F^i$ must be matched outside $G_2^i$ to a better school. Hence, inside the gadget $G_2^i$, there would be no stable and efficient matching with capacity increases at most $\eta$, as we observed before.
This implies, that if we define a truth assignment such that $X_i$ is set to True, if $F^i$ is not matched to $\overline{x_i}$ and it is set to False otherwise, then this  assignment is consistent, and for each clause, there is a true literal. 

Now, if there is a good capacity increase vector with $|\pr |_1\le \eta$, then $\lmax\le \eta$ and the same argument show the existence of a satisfying assignment $f$.
\end{proof}

If we let $\eta =3$, then it proves the first statement of the theorem, that is \maxse\ is NP-complete, even if the preference and priority lists are at most 8 long. The longest priority list is school $z_1$'s priority list in the $G_2^i$ gadgets, which is then 8 long. The longest preference list is the student $c^j$'s preference list, which is 5 long.

For the other statement, let $\eta =\hat{n}^c$, for some arbitrary constant $c>0$. Clearly, $\eta$ is polynomial in the input size, hence the reduction is still polynomial. Therefore, by \cref{claim:maxse1,claim:maxse2} it follows that it is NP-hard to decide whether the optimum is at most $3\hat{n}$ or at least $\eta =\hat{n}^c$ for both \maxse\ and \sumse\ for any $c>0$. As in $I'$ and $I''$ we have that $|U|,|E|=\mathcal{O}(\eta \hat{n})=\mathcal{O}(\hat{n}^{c+1})$, this  implies that there is no polynomial time $\mathcal{O}((|U|+|E|)^{\frac{c-1}{c+1}})$-approximation algorithm for neither \sumse\ nor \maxse\ if $P\ne NP$. As $c$ can be an arbitrarily large constant, and $\frac{c-1}{c+1}$ tends to $1$, as $c$ goes to $\infty$, this  is equivalent to saying that there is no $O((n+m)^{1-\varepsilon})$-approximation algorithm, for any $\varepsilon >0$. Also, as increasing $\eta$ does not increase the length of the students' preference lists, this  inapproximability result holds even if each preference list is at most 5 long.
\end{proof}

\section{Student-popular matchings}\label{sec:STUPOP}

In this section we consider the problems \sumstud\ and \maxstud, that is, we aim for a stable and student-popular matching. 
Recall our following simple observation from \cref{lemma:sumse-inNP}.

\begin{observation}
If there is a student-popular stable matching, then it must be the student-optimal stable matching. 
\end{observation}

We also state another simple, but crucual observation. 
\begin{observation}
If in a matching $\mu$, any student $u$ could only improve by replacing a student $u'$ at her best school, then $\mu$ is student-popular.
\end{observation}
Indeed, in this case, for any improving student $u$, there would be a disimproving student $u'$, who must leave her best school, so the sum of the votes of the students is always nonnegative for any other matching. 

\subsection{Minisum capacities}

We start by giving an example.

\begin{example}
\label{ex:studpop}
Let $\eta$ be a positive integer. 
We give an example, where there is no student-popular stable matching, even if we are allowed to increase each capacity by $\eta$. 
We have $3+2\eta$ students $u_1,u_2,e,d_1^1,\dots, d_1^{\eta },d_2^1,\dots,d_2^{\eta}$ and $3+2\eta$ schools $w_1,w_2,w_3,s_1^1,\dots,s_1^{\eta},s_2^1,\dots,s_2^{\eta}$ each with capacity one. The preferences are
\begin{center}
\begin{tabular}{ll|ll}
    $u_1:$ & $w_1$ &  $w_1:$  & $u_1\succ d_1^1\succ \cdots \succ d_1^{\eta} \succ e$ \\
    $e:$ & $w_1\succ w_2$ & $w_2:$ & $e\succ u_2\succ d_2^1\succ \cdots \succ d_2^{\eta}$ \\
    $u_2:$ & $w_2\succ w_3$ & $w_3:$ & $u_2$ \\
    $d_i^j:$ & $w_i\succ s_i^j$ & $s_i^j:$ & $d_i^j$

\end{tabular}
\end{center}
If we are only allowed to increase the capacities by at most $\eta$, then in the student-optimal stable matching, $e$ is at $w_2$ and $d_2^{\eta}$ is not at $w_2$. Hence, $d_2^{\eta}$ could get $e$'s place and improve, while $e$ could go to $w_1$ in $u_1$'s place and also improve, while only $u_1$ disimproves, meaning that the matching is not student-popular.  
However, if $e$ would be assigned outside to a better school, then the matching $M=\{ (u_1,w_1),(u_2,w_2),(s_i^j,d_i^j) \mid i\in [2],j\in [\eta ] \}$ is stable and student-popular as any student could only improve by replacing a student who is at her first place, so the number of disimproving students would be at least as much.

\end{example}

We start with the \minsum\ variant and show that \sumstud\ is NP-hard.

\begin{theorem}
\label{thm:sumstud}
\sumstud\ is NP-complete and admits no polynomial-time constant or $|\Unas |$-approximation algorithm, if $P\ne NP$. It is NP-hard even if $|\Unas |=0$.
\end{theorem}
\begin{proof} 
NP-containment follows from Lemma \ref{lemma:sumse-inNP} as in the previous problems. 

We reduce from the NP-hard \scov\ problem. As stated before in \cref{sec:SP+sum}, for any constant $d$, it is still NP-hard to approximate it within a factor of $d$. 

Let $I$ be an instance of \scov with sets $C_1\dots, C_{\m}$ and elements $\mathcal{X}=\{ 1,\dots,\en \}$ with $OPT(I)=k\le \m$. We create an instance $I'$ of \sumstud\ as follows.
\paragraph{The students and schools}
\begin{itemize}

\item[--] For each set $C_j$, $j\in [\m]$ we create a set school $c_j$ with capacity 1, along with a dummy student $d_j$.
\item[--]  For each element $i\in [\en]$, we add a gadget $G^i$ with $\eta =\m (\m+1)\en$, which is a copy of \cref{ex:studpop}. We call $e_i$ the student in $G^i$ who corresponds to $e$ in \cref{ex:studpop}.
\item[--] For each $j\in [\m]$, we add $\en$ students $u_j^1,\dots,u_j^{\en}$ along with schools $w_j^1,\dots,w_j^{\en}$ with capacity 1.  Let $U_j\coloneqq \{ u_j^1,\dots,u_j^{\en}\}$.

\end{itemize}
\paragraph{The preferences and priorities}
  The preferences and priorities are defined as follows:

  \begin{tabular}{l@{\quad}|@{\quad}l}
  \centering
      Students & Schools\\
       $e_i \colon \seqq{\mathcal{C}(e_i)} \succ G^i$ &  $c_j \colon d_j\succ \seqq{U_j}\succ \seqq{\mathcal{E}(C_j)}$ \\
    $u_j^{\ell} \colon c_j\succ w_j^{\ell}$ & $w_j^{\ell} \colon u_j^{\ell}$ \\
    $d_j \colon c_j$, &
    \end{tabular}

   \noindent where for each set $C_j$, $\mathcal{E}(C_j)$ denotes the set of students $e_i$ from the gadgets corresponding the elements in $C_j$, and for each student $e_i$, $\mathcal{C}(e_i)$ denotes the set of schools corresponding to the sets which contain~$i$.
Each school has initial capacity one.

\begin{claim}
\label{claim:sumstud-inapprox2}
For any constant $d\ge 1$, if we can find a capacity increase vector $\pr$ with $\lone \le d(k+1)\en$, then we can find a set cover using at most $d(k+1)-1$ sets.
\end{claim}
\begin{proof}
\renewcommand{\qedsymbol}{$\diamond$}
If $d>\m$, then it is easy to find a set cover using at most $d(k+1)-1\ge \m$ sets, as we can include all of the sets. 

From now on, suppose that $d\le \m$.
Let $\pr$ be a good capacity increase vector such that $|\pr |_1\le d(k+1)\en$ and let $\mu$ be a stable matching that is student-popular. For each $c_j$, $d_j$ is matched to it, because they are mutually best for each other. Hence, the original capacity 1 is filled with $d_j$ for each $c_j$. As $\eta =\m(\m +1)\en\ge d(k+1)\en$ (as we know that $d,k\le \m$), no capacity can be increased by more than $\eta$, hence by \cref{ex:studpop}, each $e_i$ must be matched to a $c_j$ school in any student-popular stable matching. 
Let $l$ denote the number of schools, where there is a student $e_i$ assigned. As for each $c_j$ all of $u_j^1,\dots,u_j^{\en}$ is better than any $e_i$ student, for each such school, all of their corresponding $u_j^{\ell}$ students must be assigned to them. As $\mu$ assings all $e_i$ students to set schools, we get that the capacity increase is at least $\en l+\en$. As $\en l+\en \le d(k+1)\en$, we get that $l\le d(k+1)-1$. As all element agents are matched in $\mu$, the $l$ sets that corresponds to these schools form a set cover.
\end{proof}
\begin{claim}
\label{claim:optinI'2}
    For the optimal capacity increase vector $\pr$ of $I'$, we have that $\lone \le (k+1)\en$.
\end{claim}
\begin{proof}
\renewcommand{\qedsymbol}{$\diamond$}
Let $C_{j_1},\dots,C_{j_{k} }$ be an optimal an optimal set cover. Define $\pr$ as follows: $\pr [w_j^{\ell}]=0$ for all $j\in [\m],\ell \in [\en]$. If $j\in \{ j_1,\dots,j_{k}\}$, $\pr [c_j]=\en+p$, where $p$ is the number of elements $i$ such that $C_j$ is the set with smallest index among $\{ j_1,\dots,j_{k} \}$ that covers $i$. Finally, $\pr [c_j]=0$ otherwise. Clearly, $|\pr |_1 \le k\en +\en$. 

    We claim that the matching $\mu$, given by the edges $\{ (d_j,c_j) \mid j\in [\m]\} \cup \{ (u_j^{\ell},w_j^{\ell}) \mid j\in [\m]\setminus \{ j_1,\dots ,j_{k}\} ,\ell \in [\en] \} \cup \{ (u_j^{\ell},c_j), (e_i,c_{l_i}) \mid j\in \{ j_1,\dots,j_{k}\} ,i,\ell \in [\en] \}$ together with the only stable matchings inside the $G^i$ gadgets without $e_i$ is stable and student-popular, where $l_i$ is the smallest index among $\{ j_1,\dots,j_{k}\}$ such that $i\in C_{l_i}$. $\mu$ is clearly feasible for $\quota +\pr$.

    Suppose there is a blocking pair $(x,y)$ to $\mu$. Student $x$ cannot be $d_j$, as $d_j$ is at her best choice. She also cannot be $u_j^{\ell}$, because for each $j\in [\m]$, either $u_j$ is at her best choice $c_j$ if $j\in \{ j_1,\dots,j_{k}\}$ or her only better school $c_j$ have capacity 1 in $\quota +\pr$, and is filled with a better student $d_j$. Student $x$ cannot be $e_i$, because again either $e_i$ is at her best school, or her better schools all have capacity 1 and filled with a better student $d_j$. This is because $C_{j_1},\dots,C_{j_{k}}$ was a set cover, and for each $e_i$ we matched her the to best school ( which is the one with the smallest index) who increased its capacity. Finally, $x$ cannot be any other student from a gadget $G^i$, because $\mu$ is stable and student-popular inside the gadgets by \cref{ex:studpop}. Hence, there is no possible choice for a blocking student $x$, contradiction. 
    To show that $\mu$ is student-popular observe that in each of the above cases for $x$ either $x$ cannot improve, or could only improve by replacing a student from her best place, so the number of improving students cannot be more than the number of disimproving ones in any matching. 
  \end{proof}

Now suppose we have a polynomial-time $d$-approximation algorithm for \sumstud, for some constant $d$. Then by \cref{claim:optinI'2}, we can find a good capacity increase vector $\pr$, with $|\pr |_1\le d(k+1)\en$.
From \cref{claim:sumstud-inapprox2}, we get that we can also find a set cover using at most $d(k+1)-1$ sets. As $\frac{d(k+1)-1}{k}\le 2d$, this  algorithm gives us a polynomial $2d$-approximation for the set cover problem, which is a contradiction, if $P\ne NP$.

As in this  reduction $|\Unas |=0$, it follows that there is no $|\Unas |$-approximation too, and also that the problem is hard even with $|\Unas |=0$.
\end{proof}

\begin{remark}
 By using \cref{ex:studpop} in the place of the edge selector gadget in the proof of \cref{sum-se:w1h-cap}, we can also show that \sumstud\ is W[1]-hard with respect to the capacity bound $\sumcap$.
\end{remark}

\subsection{Minimax capacities}

Finally, we show that \maxstud\ is solvable in polynomial time. 

We start by introducing some notations.  
Recall that we call student $u$ a \myemph{admirer} of school $w$, if $w$ is the best school for $u$.

First recall the following theorem from Malove and Sng \cite{manlove2006popular} (Theorem 1) about (one-sided) popular matchings. (There, it was stated a little differently in the context of the capacitated house allocation problem, we state its equivalent version for our case.)
\begin{theorem}{\cite{manlove2006popular}}
\label{thm:manlove-orig}
    A matching $\mu$ is student-popular if and only if: 
    \begin{enumerate}
        \item \label{manl:1} For every school $w$, that has at most $\quota [w]$ admirers, each admirer is matched to it,
        \item \label{manl:2} For every schools $w$, if it has more than $\quota [w]$ admirers, then it is filled with admirers
        \item \label{manl:3} Every student is matched to either her best school, or the best school among the ones that cannot be filled with admirers only (we allow this school to be $\emptyset$).
    \end{enumerate}
\end{theorem}

We start by providing a little simpler characterization.

\begin{theorem}
\label{thm:malove}
A matching $\mu$ is student-popular if and only if each school $w$ that can be filled with admirers only is filled with admirers of $w$ only and for each student $u$, every school better than $\mu (u)$ for $u$ is filled with admirers only.
\end{theorem}
\begin{proof}
On one hand, it is easy to see that if the conditions of \cref{thm:manlove-orig} hold, than the conditions of \cref{thm:malove} also hold. Indeed, by \ref{manl:1} and \ref{manl:2} each school that can be filled with admirers only must be filled with admirers only, and by \ref{manl:3}, if a student is not at her best school, then she must be at the best such school that cannot be filled with admirers, so all better schools must be filled with admirers.

On the other direction, suppose that the conditions of \cref{thm:malove} hold for a matching $\mu$. Then, \ref{manl:2} also holds trivially. Suppose \ref{manl:1} does not hold. Then, there is a school $w$, who is not filled with admirers, but there is an admirer $u$ who is not there. But then for $u$, there is a school better than $\mu (u)$, that is not filled with admirers, contradiction. Finally, suppose that \ref{manl:3} does not hold. By the second property, no student $u$ can be at a worse school than the best one that cannot be filled with admirers only. Suppose $u$ is at a strictly better, but not her best school $w$. Then, $w$ can be filled with admirers, but $u$ is there and not an admirer, so there is an admirer $u'$ of $w$ not there in $\mu$. However, this contradicts the first assumption for $u'$ and $\mu (u')$.

Hence we conclude that the two characterizations are equivalent.
\end{proof}

Using \cref{thm:malove}, we can show the following.

\begin{lemma}
\label{lemma:studpopstable}
An instance $I=\instance$ of \manyones\ admits a stable and student-popular matching if and only if the student proposing Gale-Shapley algorithm, which starts by all student proposing to their best school simultaneously, then the rejected students start to propose one by one; terminates without the following two cases ever happening after the first simulataneous proposal:
\begin{enumerate}
    
    \item A student $u$ gets rejected by a school $w$ that is not filled with admirers of $w$ only
    \item A student $u'$ gets dropped from a school $w$ because of a new proposal.
\end{enumerate}
\end{lemma}
\begin{proof}
First suppose that the above Gale-Shapley algorithm terminates without the two cases ever happening. Then, each student who was at her best school after the first step remains there, and each student who got rejected from their first school gets accepted to the best school that is not filled up with admirers only. Also, if a school can be filled with admirers only, then it fills up in the first step and has the same set of students in the end, so it is filled with admirers only. Hence, by \cref{thm:malove}, the student-optimal stable matching output by the algorithm is student-popular.

In the other direction, suppose that the student-optimal stable matching is student-popular, but one of the cases happen during the algorithm. If case 1 happens, then there is a student $u$ and a school $w$, such that $u$ is at a worse school than $w$ at the end, and $w$ is not filled with admirers only. By \cref{thm:malove}, this  means that the output, which is the student-optimal stable matching is not student-popular, contradiction. If case 2 happens, then at the end of the algotihm, $u'$ is at a worse place than $w$, while there is at least one student $u''$ at $w$, who is not an admirer of $w$. Hence, by \cref{thm:malove} again, the student-optimal stable matching is not student-popular, contradiction again.
\end{proof}

\paragraph{An algorithm for}\maxstud.
Consider the following algorithm. Suppose we are given a number $\maxcap$. Initially, set each capacity to $\quota [w]+\maxcap.$

\newcommand{\best}{b}
Let $\best (u)$ denote the best school of student $u$.

\textit{Phase 1:}
Let each student $u$ propose to $\best (u)$. If no school receives more proposals than its capacity, then we stop and output this  matching with the current capacity increase vector. Otherwise, the schools reject their worst students, such that they remain under their capacity. 

\textit{Phase 2:} 
The unmatched students of phase 1 start to propose in the order of their preference lists, one by one. If at any point the current student $u$ either would get rejected by a school $w$ that is not filled with admirers of $w$ only or a student $u'$ would be dropped from a school $w$ because of student $u$'s proposal, do the following. Decrease the capacity of $w$ to the number of admirers of $w$ currently at $w$ that satisfy that they are better than $u$ and also better than any student $u'$ at $w$, who is not an admirer of $w$. If the capacity gets smaller than $\quota [w]$, then we stop and output NO, there is no student-popular stable matching. Otherwise, the schools reject their worst students and we continue the proposals with the additional rejected students. If the proposal phase terminates without a school decreasing its capacity beyond its original quota, then we stop and output the obtained matching $\mu$ with the current capacity increase vector.

See the pseudocode of the algorithm in \cref{alg:maxstud}.

\begin{algorithm}
\caption{Algorithm for \maxstud}
\label{alg:maxstud}
\begin{algorithmic}
\STATE \textbf{Input}: An instance $I$ of \maxstud.
\STATE Set $\quota' [w] = \quota [w] +\maxcap$ for all $w\in W$
\STATE Set $\mu (u):= \best (u)$ for all $u\in U$
\IF{$\mu$ is feasible for $\quota '$}
\STATE \textbf{Return} $\mu$, $\pr \equiv \maxcap$
\ELSE
\STATE Each overfilled shool $w$ rejects the worst $|\mu (w)| -\quota' [w]$ of their students
\STATE Update $\mu$
\ENDIF

\WHILE{There is an unmatched student $u$}
\STATE $u$ proposes to the next school $w$ in her preference list
\IF{a student $u'$ gets rejected from $w$ because of $u$ OR $w$ is not filled with admirers but rejects $u$}
\STATE $L:=\{ u\} \cup \{ u''\in \mu (w)\mid \best (u'')\ne w\} $ (non-admirers of $w$ currently at $w$). 
\STATE Decrease $\quota '[w]$ to $|\{ u'\in \mu (w) \mid u'\succ u'' \; \forall u''\in L\} | $
\STATE $w$ rejects the worst $|\mu (w)|-\quota '[w]$ students
\STATE Update $\mu$
\IF{$\quota'[w]<\quota [w]$}
\STATE \textbf{Return} "No solution"
\ENDIF
\ENDIF
\ENDWHILE
\STATE \textbf{Return $\mu$, $\pr = \quota'-\quota$}

\end{algorithmic}

\end{algorithm}

\begin{theorem}
\label{thm:maxstud}
\maxstud\ can be solved in polynomial-time.  
\end{theorem}
\begin{proof}
We claim that \cref{alg:maxstud} outputs solution to an instance $I=(U,W,(\succ_x)_{x\in U\cup W},\quota,\maxcap)$ of \maxstud, whenever there exists a good capacity increase vector with $|r|_{\infty}\le \maxcap$. 

Suppose that the algorithm outputs a capacity increase vector $\pr$, and a matching $\mu$, which is clearly the student-optimal stable matching with respect to the capacity vector $\quota +\pr$, as it is obtained by running the student proposal Gale-Shapley algorithm with capacities $\quota + \pr$. If the algorithm terminated and output a matching, then that means that with these capacities, neither of the cases of \cref{lemma:studpopstable} happens, so by \cref{lemma:studpopstable}, the output matching is also student-popular. 

Now suppose that the algorithm terminates because a school's capacity got too small. We prove that there can be no student-popular stable matching with respect to any capacity increase vector with $|\pr |_{\infty} \le \maxcap$ (in fact, we show this for any capacity \textit{change} vectors too).

\begin{claim}
\label{claim:maxstud-induction}
    After each decrease in the algorithm, there can be no good capacity change vector with $|\pr |_{\infty}\le \maxcap$, such that any of the schools have larger capacity in $\quota +\pr$ than what they currently have in the algorithm.
\end{claim}
\begin{proof}
\renewcommand{\qedsymbol}{$\diamond$}

We show this by induction on the number of capacity decreases. Before the first decrease it is clearly true, as each school is at the maximum capacity they can be under the condition that $|\pr |_{\infty}\le \maxcap$. So suppose that we know the statement holds after the first $k$ decreases, and suppose there is a new capacity decrease.

Assume that it is because some proposing student $u$ would get rejected from a school $w$ that is not filled up with admirers. Then, the algorithm decreased $w$'s capacity to the number of admirers there who were better than both $u$ and any non-admirer student there. By induction, we know that there is no good capacity increase vector, where a school $w'\ne w$ has larger capacity, then what it has after the $k+1$-th decrease, because it remains the same. Suppose that there is one solution, where $w$'s capacity is larger than the one obtained after the decrease, but at most the one it had before the decrease (by induction it cannot be more). As the other schools' capacity is at most as large, it must hold that $u$ and all students who were at $w$ still propose to $w$ with respect to these capacities too by \cref{lem:capincrease}. Hence, there must be a non-admirer student there, because there is a non-admirer proposer such that $w$ cannot be filled with better admirer students. As $u$ got rejected before, she still has to be rejected now, because each student now at $w$ was already at $w$ before and was better than $u$. Hence, $u$ gets rejected by a school with a non-admirer student there, so by \cref{lemma:studpopstable}, there is no stable and student-popular matching with respect to these capacities, contradiction. 

Next suppose that a student $u'$ would get dropped from a school $w$, because of a proposal from a student $u$. By a similar reasoning, it is enough to show that $w$'s capacity $\quota' [w]$ cannot be larger with respect to any good capacity increase vector, then after the decrease. However, if there would be such a good capacity increase vector (where the other schools' capacities are at most as large by induction), then $u$ (and everyone in the current $\mu (w)$ before the decrease) will still propose to $w$ by \cref{lem:capincrease}. If $u'$ is there, then she gets rejected from $w$ which is a contradiction by \cref{lemma:studpopstable}. Otherwise, either there remains a student who is worse than $u$ for $w$, who gets rejected when $u$ comes, or $w$ is full with better students when $u$ proposes, so $u$ gets rejected. Also, in the second case, there must be a non-admirer student $u''$ at $w$ too, because by our assumption on $\quota [w]$, there must be a non-admirer in the top $\quota '[w]$ proposers to $w$. Hence, both these cases lead to a contradiction by \cref{lemma:studpopstable}.
\end{proof}

By \cref{claim:maxstud-induction} if a school's capacity is decreased below its original quota, then there is no nonnegative good capacity increase vector with $|\pr |_{\infty}\le \maxcap$.

\end{proof}

\begin{corollary}
The optimization version of \maxstud\ can be solved in time $\mathcal{O}(|E|\cdot |U|\cdot \log (|U|)$, where $|E|$ is the number of applications.
\end{corollary}
\begin{proof}
It is clear that for the optimal $\maxcap$ it holds that $0\le \maxcap \le |U|$. Hence, with \cref{alg:maxstud}, we can find the optimal such $\maxcap$ with binary search in $\log |U|$ iterations. Also, in each iteration, there is at most proposal along each edge in the acceptability graph and when the algorithm decreases a capacity, the decrease value can be easily computed in $\mathcal{O}(|U|)$ time, by traversing the preference list of the school.
\end{proof}

\section{Allowing decreases}\label{sec:decr}

In previous Sections we only allowed schools to increase their capacities. However, when one aims for a stable and efficient matching, it may be beneficial to allow the capacities to decrease too. For example, if one considers a clause gadget from \cref{thm:maxse}, then by decreasing $w_1$'s capacity by one, we can obtain a stable and efficient matching by matching $u_1$ to $w_2$, $u_2$ to $w_3$ and each $d_i^j$ to $s_i^j$. But as we have seen, if we allow only increases to be made, then the optimal capacity increase is at least $\eta$ (where $\eta$ can be arbitrary) both for the sum and the maximum of the changes. 

For this reason, we explore whether allowing the capacities to decrease also changes the computational complexity of the related problems.

\subsection{Stable and perfect matchings}

In the case of stable and perfect matchings, by \cref{lem:capincrease} we know that decreasing a school's capacity can never increase the size of a stable matching. Hence, in this model all results for \sumsta\ and \maxsta\ immediately transfer over. 

\subsection{Stable and efficient matchings }\label{sec:SE-decr}

 In this section we consider \sumse\ and \maxse\ in the model, where decreases are also allowed. As we have seen, here allowing decreases may indeed decrease the optimum value by a significant amount. 
 
 We show that both \maxse\ and \sumse\ remain NP-hard and inapproximable in this  case.

\begin{theorem}\label{thm:maxsedecr}
 \maxse\ and \sumse\ are NP-hard to approximate within $\mathcal{O}((|E|+|U|)^{1-\varepsilon})$ for any $\varepsilon >0$, even if we are allowed to make capacity decreases too.
\end{theorem}
\begin{proof}
The proof follows a similar construction as the proof of \cref{thm:maxse}.

Let $\eta \ge 3$ be a positive integer number specified later.

\smallskip
\noindent\textbf{Clause gadget.} We have 3 students $u_1,u_2,c$ and 3 schools $w_1,w_2,w_3$ along with $4\eta$ dummy students $d_1^1,\dots,d_1^{\eta},d_2^1,\dots,d_2^{\eta}$, $e_1^1,\dots,e_1^{\eta},e_2^1,\dots,e_2^{\eta}$ and $4\eta$ dummy schools $s_1^1,\dots,s_1^{\eta},s_2^1,\dots,s_2^{\eta}$, $t_1^1,\dots,t_1^{\eta},t_2^1,\dots,t_2^{\eta}$. The preferences are: 

  {
    \centering\begin{tabular}{ll|ll}
    $u_1:$ & $w_1\succ w_2$ &  $w_1:$  & $c\succ e_1^1\succ \cdots \succ e_1^{\eta} \succ u_1\succ d_1^1\succ \cdots \succ d_1^{\eta}\succ  u_2$ \\
    $u_2:$ & $w_1\succ w_2\succ w_3$ & $w_2:$ & $u_1\succ u_2\succ e_1^1 \succ \cdots \succ e_1^{\eta} \succ d_2^1\succ \cdots \succ d_2^{\eta} \succ c$ \\
    $c:$ & $w_2\succ w_1$ & $w_3:$ & $u_2$ \\
    $d_i^j:$ & $w_i\succ s_i^j$ & $s_i^j:$ & $d_i^j$ \\
    $e_i^j:$ & $w_i\succ t_i^j$ & $t_i^j:$ & $e_i^j$,\\
\end{tabular}
\par}

for $i\in [2], j\in [\eta ]$.
The initial capacities of all schools are one, except that $w_1$ and $w_2$ has capacity $\eta +1$. The crucial property of this  gadget is that if each school can change its capacity by at most $\eta$ , then in the student-optimal stable matching $c$ is matched to $w_1$ and either $u_1$ is matched to $w_2$ (if $w_1$ has capacity at most $\eta +1$) or $u_2$ is matched to $w_2$ (if $w_1$ has capacity more than $\eta +1$). In either case, the resulting matching is not efficient, as student $c$ would mutually like to switch places with $u_1$ or $u_2$ (the one at $w_2$).

However, if student $c$ would be matched out from the gadget to a better school, than there is a stable and efficient matching without any capacity changes needed, $\mu =\{ (u_1,w_1),(u_2,w_2),(e_i^j,w_i)$, $(d_i^j,s_i^j)\mid i\in [2],j\in [\eta  ]\}$.

\smallskip
\noindent\textbf{Variable gadget.}
We have 2 schools $x,\overline{x}$ with 8 students $T$, $F$, $c_1,c_2,\overline{c}_1,\overline{c}_2,y,\overline{y}$ along with $2\eta$ dummy students $f_1,\dots,f_{\eta},g_1,\dots,g_{\eta}$ and $2\eta$ dummy schools $p_1,\dots,p_{\eta},q_1,\dots,q_{\eta}$. The preferences are:

{
    \centering\begin{tabular}{ll|ll}
    $T:$ & $x\succ \overline{x} $ &  $x:$  & $y\succ F\succ c_1\succ c_2\succ f_1\succ \cdots \succ f_{\eta} \succ T$ \\
    $F:$ & $\overline{x} \succ x $ & $\overline{x}:$ & $\overline{y}\succ T\succ \overline{c}_1\succ \overline{c}_2\succ g_1 \succ \cdots \succ g_{\eta} \succ F$ \\
    $f_i:$ & $x\succ p_i$ & $p_i:$ & $f_i$ \\
    $g_i:$ & $\overline{x}\succ q_i$ & $q_i:$ & $g_i$, \\
    $c_1,c_2,y:$ & $x$ &  & \\
    $\overline{c}_1,\overline{c}_2,\overline{y}:$ &$\overline{x}$ & &\\

\end{tabular}
\par}
The initial capacity is 1 for all schools.
The important property of this gadget is that if both $x$ and $\overline{x}$ have a $c^j$ or $\overline{c}^j$ student, then there is no stable and efficient matching with at most $\eta$ capacity change. This holds, because if each school can only change its capacity by at most $\eta$ and both $x$ and $\overline{x}$ have capacity at least two, (and they have capacity at most $\eta +1$), then $T$ must be at $\overline{x}$ and $F$ must be at $x$ in the student-optimal stable matching, hence they would mutually like to switch schools.

However, if one of $x,\overline{x}$, (say $\overline{x}$) has capacity 0, then there is a stable and efficient matching $\mu$, where one or two of $c_1,c_2$ can be at $x$, $x$ has capacity $2+|\{ c_1,c_2\} \cap \mu (x)|$ and no other school changes capacity, namely $\mu =\{ (y,x), (F,x),(c_j,x),(f_i,p_i),(g_i,q_i)\mid i\in [\eta] \}$, where $j\in \{ 1\}$, $\{ 2\}$ or $\{ 1,2\}$, depending on which students from $\{ c_1,c_2\} $ get matched to $x$.

\smallskip
\noindent \textbf{The reduction.}
Now we provide our reduction from \rrsat.

Let $I=\varphi$ be an instance of \rrsat, with clauses $C_1,\dots,C_{\hat{m}}$ and variables $X_1,\dots,X_{\hat{n}}$. Note that $3\hat{m}=4\hat{n}$. We create an instance $I'$ of \maxse\ and an instance $I''$ of \sumse\ as follows.

\begin{compactitem}[--]
    \item For each clause $C_j$, $j\in [\hat{m}]$, we create a clause gadget $G^j$. Let us denote the student $c$ in gadget $G^j$ by $c^j$ and let us call these students the \textit{clause students}.
    \item For each variable $X_i$, we create a variable gadget $H^i$. Let us call the $T$ and $F$ students in $H^i$ by $T^i$ and $F^i$ and the schools $x$ and $\overline{x}$ by $x^i$ and $\overline{x}^i$. Let us call these students and schools the \textit{boolean students} and \textit{literal schools} respectively.
\end{compactitem}
As $3\hat{m}=4\hat{n}$, we have $\mathcal{O}(\eta \hat{n})$ students and schools.
 Finally, we identity the clause student $c^j$ from the clause gadget $G^j$ with one of the students $c^i_1,c^i_2,\overline{c}^i_1,\overline{c}^i_2$ from a variable gadget $H^i$, whenever $X_i$ or $\overline{X}_i$ appears in $C_j$. We identify it with $c^i_1$, if this is the first appearence of $X_i$, with $c^i_2$, if it is the second appearence of $X_i$, with $\overline{c}^i_1$, if it is the first appearence of $\overline{X}_i$ and with $\overline{c}^i_2$, if it is the second appearence of $\overline{X}_i$. (So for example if $C_1=(X_1\vee X_2\vee \overline{X_3} )$, and $X_1,\overline{X_3}$ are first appearences, and $X_2$ is a second appearence, then we identify the students $c^1$ with $c_1^1$, $c_2^2$ and $\overline{c}_1^3$).
 We extend the preferences of these students, such that they rank their adjacent $x^i$ or $\overline{x}^i$ schools first in the order the corresponding literals appear in $C_j$, and the schools from their clause gadget after, in their original order.

\begin{claim}
\label{claim:maxsedecr1}
If there exists a satisfying assignment, then there is a good capacity change vector $\pr$ with $\lone \le 2\hat{n}+ \hat{m}$ and $\lmax \le 3$.
\end{claim}
\begin{proof}
\renewcommand{\qedsymbol}{$\diamond$}
Let $f$ be a satisfying assignment. Let $\mu$ be the following matching.
If $X_i$ is set to True, then we match $F^i$ and $y_i$ to school $x^i$ and noone to $\overline{x}_i$ (it will have 0 capacity), otherwise we match $\overline{y}_i$ and $T^i$ to school $\overline{x}^i$ and noone to $x_i$. Then, for each clause $C_j$, let $\ell_{i_j}$ be the first literal in $C_j$ that is true and match student $c^j$ from $G^j$ to the literal school of $\ell_{i_j}$. Finally match $u_1$ to $w_1$, $u_2$ to $w_2$ in each clause gadget $G^i$ and match each dummy student to her dummy school, except that the students $e_j^1,\dots,e_j^{\eta}$ are matched to $w_j$ for $j\in \{ 1,2\}$ in each clause gadget $G^i$. Next, we describe $\pr$. For each $i\in [\hat{n}]$, if $X_i$ is set to True, then we define $\pr [x_i]$ to be one plus the number of clause students matched there and $\pr [\overline{x}_i]=-1$, otherwise we let $\pr [x_i]=-1$ and $\pr [\overline{x}_i]$ to be one plus the number of clause students matched there. For all other schools $s$, $\pr [s]=0$. Clearly, $|\pr|_1=2\hat{n} +\hat{m}$ and $\lmax \le 3$, because each literal school has at most two clause students. 

It is straightforward to verify that $\mu$ is feasible for $\quota +\pr$.

As $f$ was a satisfying assignment, for each $j\in [\hat{m}]$, $c^j$ is matched outside of the gadget $G_1^j$ to a better school, hence as we have seen before, the rest of the students inside are matched in a stable and efficient way. As $f$ is consistent, for each $i\in [\hat{n}]$, exactly one of $x^i$ and $\overline{x}^i$ has nonzero capacity. 

To see that $\mu $ is stable it is straighforward to verify that for each student who is not at her best school, all better schools are filled with better students or have zero capacity in $\quota +\pr$. 

To see that $\mu$ is efficient, observe that each school $x^i,\overline{x}^i$ either has zero capacity, or every student there is at her best school with nonzero capacity. The dummy schools are not envied by anyone. Finally, each school $w_1,w_2,w_3$ in a gadget $G^i$ satisfies that all students there are at their best school with nonzero capacity, except the $u_2$ students, who envy $w_1$. However, for a $u_2$ student to improve, she would have to replace someone from $w_1$ and the replaced student could not improve, as she is at her best school. Hence, no student can improve by going to a different school, without making someone worse off. It follows that $\mu $ is efficient.
\end{proof}

\begin{claim}
\label{claim:maxsedecr2}
If there is a good capacity change vector $\pr$ with $|\pr|_{\infty}\le \eta$ or $|\pr |_1\le \eta$, then there is a satisfying assignment.
\end{claim}

  \begin{proof}
\renewcommand{\qedsymbol}{$\diamond$}
Suppose there is a good capacity change vector $\pr$ with $\lmax\le \eta$ and a stable and efficient matching $\mu$ with respect to capacities $\quota +\pr$. Then, it must hold that for each gadget $G^j$, $c^j$ is matched to a literal school, as otherwise there is no stable and efficient matchings inside $G^j$ with capacity change at most $\eta$. Furthermore, for each $i\in [\hat{n}]$ at most one of $x^i$,$\overline{x}^i$ at most one has a $c^j$ clause student, because otherwise there would be no stable and efficient matching with at most $\eta$ increase, as $T^i$ and $F^i$ would envy each other.
This implies, that if we define a truth assignment $f$ such that $X_i$ is set to True, if $x^i$ has a clause student $c^j$ in $\mu$ and it is set to False otherwise, then this  assignment is consistent, and for each clause, there is at least one true literal inside. 

Now, if there is a good capacity increase vector with $|\pr |_1\le \eta$, then $\lmax\le \eta$ and the same argument shows the existence of a satisfying assignment $f$.
\end{proof}

If we let $\eta =3$, then it proves that \maxse\ is NP-complete, even if the preference and priority list lengths are bounded by constant and decreases are allowed. 

To see inapproximability and hardness of \sumse\ too, let $\eta =(2\hat{n}+\hat{m})^c$, for some arbitrary constant $c>0$. Clearly, $\eta$ is polynomial in the input size, hence the reduction is still polynomial. Therefore, by \cref{claim:maxsedecr1,claim:maxsedecr2} it follows that it is NP-hard to decide whether the optimum is at most $2\hat{n}+\hat{m}=\mathcal{O}(\hat{n})$ or at least $\eta =(2\hat{n}+\hat{m})^c=\mathcal{O}(\hat{n}^c)$ for both \maxse\ and \sumse\ for any $c>0$. As in $I'$ and $I''$ we have that $|U|,|E|=\mathcal{O}(\eta \hat{n})=\mathcal{O}(\hat{n}^{c+1})$, this  implies that there is no polynomial time $\mathcal{O}((|U|+|E|)^{\frac{c-1}{c+1}})$-approximation algorithm for neither \sumse\ nor \maxse\ if $P\ne NP$. As $c$ can be an arbitrarily large constant, and $\frac{c-1}{c+1}$ tends to $1$, as $c$ goes to $\infty$, this  is equivalent to saying that there is no $O((|U|+|E|)^{1-\varepsilon})$-approximation algorithm, for any $\varepsilon >0$. Also, as increasing $\eta$ does not increase the length of the students' preference lists, this  inapproximability result holds even if each preference list is at most 5 long (the clause students have the longest preferences, which are 5 long).
\end{proof}

\subsection{Stable and student-popular matchings with decreases allowed}\label{sec:STUPOP-decr}

Similarly as with stable and efficient matchings, it can be beneficial to allow decreases in the capacities to obtain stable and student-popular matchings too. In \cref{ex:studpop}, if we decrease $w_1$'s capacity by 1 to 0, then the matching $\mu =\{ (e,w_2),(u_2,w_2),(d_i^j,s_i^j)\mid i\in [2],j\in [\eta ]\}$ is stable and student-popular, whereas by only increases at least $\eta$ change was necessary (where $\eta$ could be chosen arbitrarily). Hence, by allowing decreases, the optimum value can be $\mathcal{O}(|U|)$ times less.

We extend both the hardness result for \sumstud\ and the algorithm for \maxstud\ for this case. 

\begin{theorem}
\label{thm:studpop-sum-decr}
\sumstud\ is NP-hard and admits no polynomial-time constant or $|\Unas |$-approximation algorithm, even if capacity decreases are allowed.
\end{theorem}
\begin{proof} 
We reduce from the NP-hard \scov\ problem. Let $d$ be any constant.
Let $I$ be an instance of \scov\ with $OPT(I)=k\le \hat{m}$, where $\hat{m}$ is the number of sets. We create an instance $I'$ of \sumstud\ as follows.
\paragraph{The element gadget} First we describe the element gadgets we are going to use.
Let $\eta$ be a positive integer. 
We have $3+4\eta$ students $u_1,u_2,e,d_1^1,\dots, d_1^{2\eta },d_2^1,\dots,d_2^{2\eta}$ and $3+4\eta$ schools: $w_1,w_2$ with capacity $\eta +1$ and $w_3,s_1^1,\dots,s_1^{2\eta},s_2^1,\dots,s_2^{2\eta}$ with capacity 1. The preferences are
\begin{center}
\begin{tabular}{ll|ll}
    $u_1:$ & $w_1$ &  $w_1:$  & $u_1\succ d_1^1\succ \cdots \succ d_1^{2\eta} \succ e$ \\
    $e:$ & $w_1\succ w_2$ & $w_2:$ & $e  \succ u_2\succ d_2^1\succ \cdots \succ d_2^{2\eta}$ \\
    $u_2:$ & $w_2\succ w_3$ & $w_3:$ & $u_2$ \\
    $d_i^j:$ & $w_i\succ s_i^j$ & $s_i^j:$ & $d_i^j$ \\
   \\

\end{tabular}
\end{center}
If we are only allowed to change the capacities by at most $\eta$, then in the student-optimal stable matching, $e$ is at $w_2$ and $d_2^{2\eta}$ is not at $w_2$. Hence, $d_2^{2\eta}$ could get $e$'s place and improve, while $e$ could go to $w_1$ in $u_1$'s place and also improve, while only $u_1$ disimproves, meaning that the matching is not student-popular.  
However, if $e$ would be assigned outside to a better school, then the matching $\mu =\{ (u_1,w_1),(u_2,w_2),(d_i^j,w_i),(d_i^l,s_i^l)\mid i\in [2],j\in [\eta ],l\in \{ \eta +1,\dots,2\eta \} \}$ is stable and student-popular as any student could only improve by replacing out a student who is at her first place.

 Now we proceed with describing the reduction.
\paragraph{The students and schools}
\begin{itemize}

\item[--] For each set $C_j$, $j\in [\hat{m}]$ we create a set school $c_j$ with capacity 1, along with a dummy student $d_j$.
\item[--]  For each element $i\in [\hat{n}]$, we add an element gadget $G^i$ with $\eta =\hat{m}(\hat{m}+1)\hat{n}$. We call $e^i$ the student in $G^i$ who corresponds to $e$.
\item[--] For each $j\in [\hat{m}]$, we add $\hat{n}$ students $u_j^1,\dots,u_j^{\hat{n}}$ along with schools $w_j^1,\dots,w_j^{\hat{n}}$ with capacity 1.

\end{itemize}
\paragraph{The preferences and priorities}
  The preferences and priorities are defined as follows:

  \begin{tabular}{l@{\quad}|@{\quad}l}
  \centering
      Students & Schools\\
       $e^i \colon \seqq{\mathcal{C}( e^i)} \succ w_1^i \succ w_2^i$ &  $c_j \colon d_j\succ \seqq{U_j}\succ \seqq{\mathcal{E}(C_j)}$ \\
    $u_j^{\ell} \colon c_j\succ w_j^{\ell}$ & $w_j^{\ell} \colon u_j^{\ell}$ \\
    $d_j \colon c_j$, &
    \end{tabular}

   \noindent where for each set $C_j$, $\mathcal{E}(C_j)$ denotes the set of students $ e^i$ from the gadgets corresponding the elements in $C_j$, and for each student $ e^i$, $\mathcal{C}( e^i)$ denotes the set of schools corresponding to the sets which contain~$i$ and $U_j\coloneqq \{ u_j^1,\dots,u_j^{\hat{n}}\}$.
Each school has initial capacity one.

\begin{claim}
\label{claim:sumstud-inapp-decr1}
For each constant $d\ge 1$, if we can find a good capacity change vector $\pr$ with $|\pr |_1\le d(k+1)\hat{n}$ then we can find a set cover using at most $d(k+1)-1$ sets.
\end{claim}
\begin{proof}
\renewcommand{\qedsymbol}{$\diamond$}
If $d>\hat{m}$, then it is trivial to find a set cover using at most $d(k+1)-1\ge \hat{m}$ sets. 

From now on, suppose that $d\le \hat{m}$.
Let $\pr$ be a good capacity change vector such that $|\pr |_1\le d(k+1)\hat{n}$ and let $\mu$ be a stable matching that is student-popular. For each $c_j$, $d_j$ is matched to it, because they are mutually best for each other, except only when $c_j$ decreases its capacity to 0. As $\eta =\hat{m}(\hat{m}+1)\hat{n}\ge d(k+1)\hat{n}$, no capacity can be increased by more than $\eta$, hence by the observed property of the element gadget, each $e^i$ must be matched to a $c_j$ school in any student-popular stable matching. 
Let $l$ denote the number of $c_j$ schools, where there is a student $e^i$ assigned. As for each $c_j$ all of $u_j^1,\dots,u_j^{\hat{n}}$ are better than any $e^i$ student, for each such school, all of their corresponding $u_j^{\ell}$ students must be assigned to them. As $\mu$ assings all $e^i$ students to set schools, we get that the capacity change is at least $\hat{n}l+\hat{n}$. As $\hat{n}l+\hat{n}\le d(k+1)\hat{n}$, we get that $l\le d(k+1)-1$. As all element agents are matched in $\mu$, the $l$ sets that corresponds to these schools form a set cover.
\end{proof}

\begin{claim}
    We have $OPT(I')\le k\hat{n}+\hat{n}$. 
\end{claim}
\begin{proof}
\renewcommand{\qedsymbol}{$\diamond$}
Let $C_{j_1},\dots,C_{j_{k} }$ be an optimal set cover using $k=OPT(I)$ sets. Define $\pr$ as follows: $\pr [w_j^{\ell}]=0$ for all $j\in [\hat{m}],\ell \in [\hat{n}]$. If $j\in \{ j_1,\dots,j_{k}\}$, $\pr [c_j]=\hat{n}+p$, where $p$ is the number of elements $i$ such that $C_j$ is the set with smallest index among $\{ j_1,\dots,j_{k} \}$ that covers $i$. Finally, $\pr [c_j]=0$ otherwise. Clearly, $|\pr |_1\le k\hat{n}+\hat{n}$. 

    We claim that the matching $\mu$, given by the edges $\{ (d_j,c_j) \mid j\in [\hat{m}]\} \cup \{ (u_j^{\ell},w_j^{\ell}) \mid j\in [\hat{m}]\setminus \{ j_1,\dots ,j_{k}\} ,\ell \in [\hat{n}] \} \cup \{ (u_j^{\ell},c_j), (e^i,c_{l_i}) \mid j\in \{ j_1,\dots,j_{k}\} ,i,\ell \in [\hat{n}] \}$ together with the only stable matchings inside the $G^i$ gadgets without $e^i$ is stable and student-popular, where $l_i$ is the smallest index among $\{ j_1,\dots,j_{k}\}$ such that $i\in C_{l_i}$. $\mu$ is clearly feasible for $\quota +\pr$.

    Suppose there is a blocking pair $(x,y)$ to $\mu$. Student $x$ cannot be $d_j$, as $d_j$ is at her best choice. She also cannot be $u_j^{\ell}$, because for each $j\in [\hat{m}]$, either $u_j$ is at her best choice $c_j$ if $j\in \{ j_1,\dots,j_{k}\}$ or her only better school $c_j$ have capacity 1 in $\quota +\pr$, and is filled with a better student $d_j$. She cannot be $ e^i$, because again either $ e^i$ is at her best school, or her better schools all have capacity 1 and filled with a better student $d_j$. This  is because $C_{j_1},\dots,C_{j_{k}}$ was a set cover, and for each $ e^i$ we matched her the to best school ( which is the one with smallest index) who increased its capacity. Finally, $x$ cannot be any other student from a gadget $G^i$, because $\mu$ is stable and student-popular inside the gadgets. Hence, there is no possible choice for a blocking student $x$, contradiction. 
    To show that $\mu$ is student-popular observe that in each of the above cases for $x$ either $x$ cannot improve, or could only improve by replacing a student from her best place, so the number of improving students cannot be more than the number of disimproving ones in any matching. 
  \end{proof}

Now suppose we have a polynomial-time $d$-approximation algorithm for \sumstud, for some constant $d$. Then, we can decide whether there is a good capacity change vector $\pr$, with $|\pr |_1\le d(k+1)\hat{n}$.
From \cref{claim:sumstud-inapp-decr1}, we get that we can also decide whether there is a set cover using at most $d(k+1)-1$ sets. As $\frac{d(k+1)-1}{k}\le 2d$, this  algorithm gives us a polynomial $2d$-approximation for the set cover problem, which is a contradiction, if $P\ne NP$.

As in this reduction $|\Unas |=0$, it follows that there is no $|\Unas |$-approximation either, and also that the problem is hard even with $|\Unas |=0$.
\end{proof}

Finally, we show that the \cref{alg:maxstud} for \maxstud\ easily extends to the case when decreases are allowed too. We only change the stopping criteria for a NO answer of the algorithm, to only stop and output NO, if a school's capacity would be decreased below $\quota [w]-\maxcap$ (instead of $\quota [w]$).

\begin{theorem}
    \maxstud\ can be solved in polynomial time, even if decreases are allowed. 
\end{theorem}
\begin{proof}
The proof follows again from \cref{claim:maxstud-induction}.
\end{proof}

\section{Future Work}\label{sec:conclude}

For future work, one could investigate parameterized complexity for other parameters such as the number~$m$ of schools.
For constant number~$m$ of schools, it is easy to see that all problems can be done in polynomial time by guessing for each school the capacity increase and checking whether there exists a solution for the guessed capacities.
Secondly, one could look at stable matching with maximum cardinality instead of perfectness. It would be interesting to see whether the algorithmic results for \sumsta\ and \sumsta\ transfer to this  case. 

Finally, 
one could look at other objectives such as stable and popular matching (not student-popular like in this paper).

\section{Acknowledgements}
Gergely Csáji  acknowledges the financial support by the Hungarian Scientific Research Fund, OTKA, Grant No. K143858, by the Ministry of Culture and Innovation of Hungary from the National Research, Development and Innovation fund, financed under the KDP-2023 funding scheme (grant number C2258525) and  by the Hungarian Academy of Sciences, Momentum Grant No. LP2021-1/2021. Jiehua Chen has been funded by the Vienna Science and Technology Fund (WWTF) [10.47379/VRG18012].

\bibliographystyle{abbrv} 

\bibliography{main}  
\iflong
\clearpage

\begin{table}[t!]
  \centering
  \LARGE \textbf{\appendixtitle}
\end{table}
\bigskip

\appendix

\appendixtext

\fi
\clearpage

\end{document}